\documentclass[atmp]{ipart_v1}

\Vol{xx}
\Issue{x}
\Year{20xx}
\firstpage{1}

\usepackage{t1enc}
\usepackage[latin1]{inputenc}
\usepackage[english]{babel}

\usepackage{amsthm}
\usepackage{yfonts}

\usepackage{bbm}
\usepackage{bm}
\usepackage{mathrsfs}

\usepackage{slashed}

\usepackage{graphicx}
\usepackage{amsmath}

\usepackage{fancyhdr}
\usepackage{afterpage}
\usepackage{amssymb,amscd,amsthm,amsxtra}
\usepackage{latexsym}
\usepackage[dvips]{epsfig}
\usepackage{stmaryrd,mathtools}
\newcommand{\R}{\mathbb{R}}     
\newcommand{\M}{\mathcal{M}}    
\newcommand{\J}{\mathcal{J}}
\newcommand{\PP}{\mathcal{P}}
\newcommand{\HH}{\mathcal{H}}

\newcommand{\f}{\textit{graph}{f}}
\newcommand{\ov}[1]{\overline{#1}}                         
   
\newcommand{\be}[0]{\begin{equation}}
\newcommand{\ee}[0]{\end{equation}}

\numberwithin{equation}{section}

\theoremstyle{plain}
\newtheorem{theorem}{Theorem}[section]
\newtheorem{lemma}[theorem]{Lemma}
\newtheorem{proposition}[theorem]{Proposition}
\newtheorem{remark}[theorem]{Remark}
\newtheorem{corollary}[theorem]{Corollary}

\begin{document}

\title[Blowup rate control for solution of Jang's equation and its application]{Blowup rate control for solution of Jang's equation and its application on Penrose inequality}

\author[Wenhua Yu]{Wenhua Yu}

\begin{abstract}
We prove that the blowup term of a  blowup solution of Jang's equation on an initial data set $ (\M,g,k) $ near an arbitrary strictly stable MOTS $ \Sigma $ is exactly $ -\frac{1}{\sqrt{\lambda}}\log \tau $, where $ \tau $ is the distance from  $ \Sigma $ and $ \lambda $ is the principal eigenvalue of the  MOTS stability operator of $ \Sigma $. 
We also prove that the gradient of the solution is of order $ \tau^{-1} $. Moreover, we apply these results to get a Penrose-like inequality under additional assumptions.
\end{abstract}

\maketitle

\section{Acknowledgement}
This is my 2019 PhD dissertation thesis (DOI: 10.7916/d8-avnq-g588). I would like to thank my dissertation advisor, Professor Mu-Tao Wang, for his guidance and support throughout the time of my research. Also I want to thank Pei-ken Hung for helpful discussions and suggestions.

\newtheorem*{foliationA}{Foliation A}
\newtheorem*{foliationB}{Foliation B}
\newtheorem*{foliationC}{Foliation C}

\section{Preliminary and Main Results}\label{introduciton}
In this paper, we specify the relationship between the  blowup term of any  blowup solution of Jang's equation at a strictly stable marginally outer trapped surface  (MOTS) and the principal eigenvalue of the stability operator. We also specify the blowup order of the gradient of the solution.   We then apply these results to the Jang's slice constructed by J. Metzger (c.f. \cite{metzger},\cite{jek}) to get a Penrose-like inequality.\par

Jang's equation is a quasilinear elliptic equation on an initial data set  and it was proposed by P.  Jang in \cite{jangs}. An important application of this equation is the proof of Positive Mass Theorem on a  general initial data set by R. Schoen and S.-T. Yau in \cite{yau}. 
As in  \cite{yau}, we assume $ (\M,g,k) $ is an oriented 3-dim initial data set, with $ g $ the Riemannian metric,   $ k $ the second fundamental form, $ \mu $ the local mass density, $ J^i $ the local current density, satisfying the constraint equations
\begin{align*}
\mu=&\frac{1}{2}(R-k^{ij}k_{ij}+(tr_g k)^2)\\
J^i=&\nabla_j(p^{ij}-g^{ij}tr_g k)
\end{align*}

The initial data set $ (\M,g,k) $ is said to be asymptotically flat if the complement of some compact subset of  $ \M $
consists of finitely many connected components $ \M_1, . . . , \M_m $, called the ends,
each one diffeomorphic to $ \R^3-\overline{ B_1(0)} $, and such that in the corresponding
coordinate systems the metric tensor $ g_{ij} $ converges to the Euclidean metric
$ \delta_{ij} $ and the second fundamental form tensor $k_{ij} $ to zero. More precisely, we require that $ |g_{ij} - \delta_{ij}| + |x||\partial_k g_{ij} | = O(|x|^{- q}) $ and $ k_{ij}= O(|x|^{-q-1}) $ as $ |x| \to\infty$ for some $ q>\dfrac{1}{2} $.
We ask in addition that for some $ \beta > 2 $, $ tr_g(k) = g^{ij}k_{ij} = O(|x|^{-\beta}) $ as $ x \to \infty $. This last condition is imposed so that certain barriers for the Jang's equation can be constructed far out in the asymptotically flat ends. With each end $ \M_k $ we associate a total mass $ m_k $ defined by the flux integral:
$$m_k=\frac{1}{16\pi}\int_{\infty}\sum_{i,j}(g_{ij,j}-g_{jj,i})d\sigma_i$$
which is the limit of surface integrals taken over large two spheres in $ \M_k $. This number $ m_k $ is called the ADM mass of $ \M_k $, c.f. \cite{adm}. In this formulation, the general Positive Mass Theorem  states that, if  $ (\M,g,k) $ is a complete  asymptotically flat initial data set and satisfies the dominant energy condition $ \mu>|J| $, then the ADM mass of each end is non-negative. In \cite{yau}, R. Schoen and S.-T. Yau proved this theorem by reducing it to the Riemannian Positive Mass Theorem with the help of Jang's equation,  which states that for a 3-dim complete asymptotically flat Riemannain manifold  with non-negative scalar curvature,  the ADM  mass of each end  is non-negative.\par

An important tool applied in this reduction strategy is Jang's equation.
Consider the graph of a function $f \in C^2(\M)$ as a hypersurface of the product manifold $(\M \times \R, g + dt^2)$. 
Jang's equation is then:
\begin{equation} \J[f] = \HH [f] - \PP [f] = 0, \label{jang} \end{equation}
where $\HH[f]$ denotes the mean curvature of $\f \subset \M\times\R$ computed with respect to its downward-pointing unit normal, and $\PP[f] $ denotes the trace of $ k $ with respect to the induced metric on ${\f}$, after  extending $k$ trivially along the $\R$-factor.
Thus $ \HH[f],\PP[f] $ are given by:
\begin{align*}
\HH[f]=&(g^{ij}-\frac{f^i f^j}{1+|\nabla f|^2})\frac{\nabla_i\nabla_j f}{\sqrt{1+|\nabla f|^2}}\\
\PP[f]=&(g^{ij}-\frac{f^i f^j}{1+|\nabla f|^2})k_{ij}
\end{align*}

In \cite{yau}, in order to construct a solution for Jang's equation Eq.\eqref{jang},  it was proved that the auxiliary equation: 
\begin{equation}\label{aux}
\HH[f]-\PP[f]=tf
\end{equation}
has a solution $ f_t\in B^{2,\beta} $ when $ t>0 $.  
Here, the weighted H\"{o}lder space $ B^{2,\beta} $ for $ \beta\in(0,1) $ is defined as the space for all $ f $ such that $ |f|_{2,\beta} $ is finite, where
\begin{align*}
|f|_{2,\beta}=&\sup_{\M}(r^{\beta}(x)|f(x)|+r^{1+\beta}(x)|\nabla f(x)|\\ &+r^{2+\beta}(x)|\nabla^2 f(x)|+r^{2+2\beta}(x)|\nabla^2 f |_{\beta,x})\\
|\nabla^2 f |_{\beta,x}=&\sup_{x_1,x_2\in B_{\frac{r(x)}{2}}(x)} \frac{|\nabla^2 f(x_1)-\nabla^2 f(x_2)|}{dist(x_1,x_2)^{\beta}}
\end{align*}
where the weight function $ r(x) $ satisfies $ r(x)\geq 1 $ in $ \M $, and $ r(x)=|x| $ on each end $ \M_k $.\par

The solution $ f_t $ for the  auxiliary equation Eq.\eqref{aux} satisfies the a priori estimate:
\begin{equation}
\sup_{\M} t|f|\leq \mu_1,\mbox{ }\mbox{ }\mbox{ }\sup_{\M} t|\nabla f|\leq \mu_2
\end{equation}
where $ \mu_1,\mu_2 $ are constants depending only on the initial data. With this gradient bound and the a priori estimate \cite[Proposition 2]{yau},  a smooth solution $ f $ of Jang's equation Eq.\eqref{jang} can be found, as the  limit function of the sequence $ \{f_{t_i}\} $ for a sequence $ \{t_i\} $ converging to zero. This solution of Jang's equation might blow up or down at some apparent horizons $ \Sigma $. By apparent horizon we mean   a closed surface $ \Sigma $ such that one of $ \theta^+[\Sigma] $ or $ \theta^-[\Sigma] $ vanishes on $ \Sigma $, where 
\[ \theta^\pm = \text{tr}_g k - k(\nu, \nu) \pm H, \]
is called \textit{inner} and \textit{outer expansions}. Here $\nu$ is the unit normal on $ \Sigma $ pointing into $ \M $ and $H$ is the mean curvature of $ \Sigma $ in $ \M $ with respect to $\nu$. Also by blowup  we mean, that outside from an apparent horizon $ \Sigma $ the function $ f $ is such that $ \f $ is a smooth submanifold of $ \M\times\R $ with a cylindrical end converging to $\Sigma \times \R$.\par

After the solution of Jang's equation has been constructed, it is applied to the proof of  the general Positive Mass Theorem. It helps reduce  the general Positive Mass Theorem to the Riemannian Positive Mass Theorem. We review the procedure of this reduction strategy here. For each apparent horizon $ \Sigma $ where the solution $ f $ blows up or down, denote $ U $ a neighborhood of $ \Sigma $, then in \cite{yau} each of these cylindrical end of $ \f $ is slightly deformed in $ U\times(T, \infty) $ or $ U\times(-\infty,-T) $ so that $ \f $ coincides with $ \Sigma\times\R $ in $ U\times(T, \infty) $ or $ U\times(-\infty,-T) $. These  cylindrical ends are then closed up by a conformal factor $ \psi $ which exponentially goes to zero on these ends. After this step, the metric on each of these cylindrical ends is conformally changed to a metric which is uniformly equivalent to the flat metric in a punctured ball. Then, a conformal factor $ u $ on $ \M $ for this deformed metric $ \tilde{g} $ can be found such that $ u^4\tilde{g} $ has zero scalar curvature. The theorem  is then proved because of the Riemannian Positive Mass Theorem and the fact that the ADM mass of $ u^4\tilde{g} $ is not greater than that of $ g $.\par

It is well known that in the time symmetric case when black holes are present, there is a positive lower bound for the mass, in terms of the surface area of the black holes. The Riemannian Penrose inequality may be thought of as a refinement of the Riemannian Positive Mass Theorem in the presence of horizons. It relates the total ADM mass (of a chosen
end) $m$ to the area $A$ of its outermost minimal area enclosure. The inequality states that: 
\begin{equation}
m\geq\sqrt{\frac{A}{16\pi}}.
\end{equation}
for an  asymptotically flat Riemannian manifold with nonnegative scalar curvature. It is conjectured that Penrose inequality also holds for general initial data sets, when horizons are replaced by MOTSs. It is natural to consider applying Jang's equation to prove the general Penrose inequality, because until now this is the only equation that can locate the apparent horizons. However,  in the proof of the Positive Mass theorem \cite{yau}, the information of the apparent horizon $ \Sigma $ is lost after the  cylindrical end being deformed  to coincide  with $ \Sigma\times\R $ in $ U\times(T, \infty) $ or $ U\times(-\infty,-T) $. For example, in Schwarzschild spacetime, we can construct a blowup solution for Jang's equation on its time symmetric slice. Then if we applied the same  reduction techniques as in \cite{yau} to this blowup solution,  the cylindrical end will be closed up and the corresponding metric will then be conformally deformed to the flat metric on punctured $ \R^3 $. If one want to extract some information about $ \Sigma $, then the original blowup end of the solution is needed.  Due to this issue, H.  Bray and M.   Khuri proposed in \cite{coupled} a way to prove the Penrose inequality for general initial data sets by using the solution of generalized Jang's equation. Also in the same paper a solution of  generalized Jang's equation in spherically symmetrical setting was found and was used to prove Penrose inequality in this case. However,   the existence of the solution of generalized Jang's equations is still not clear in the general case. For the original equation proposed by P.   Jang,  J. Metzger was able to construct in \cite{metzger} a solution which blows up at the outermost MOTS. In this paper, our intuition is to obtain a sharp estimate of blowup rate of the solution of the original Jang's equation Eq.\eqref{jang} at one blowup end, and thus capture some information of the corresponding horizon $ \Sigma $ in form of a Penrose-like inequality.  \par

There has been a lot of research about Jang's equation and its blowup solution, c.f. \cite{yau},\cite{metzger},\cite{jek},\cite{han},\cite{williams}. In these papers, various existence theorems of Jang's equation were proved, with different constraints, such as asymptotically decaying and blowup at horizon.  Blowup control estimates of these solutions were also given in these papers. The uniqueness of asymptotically decaying solution of Jang's equation which blowup at outermost MOTS in the time symmetric (i.e. $ k=0 $) case was proved in \cite{jek}.  We review these results in Section \ref{existing}. \par

In this paper, we study the blowup rate of any possible  blowup solution of Jang's equation \eqref{jang} near a strictly stable MOTS $ \Sigma $ on a general initial data set $ (\M,g,k) $. We prove in Section \ref{barrier} that the blowup term of  any such kind of solution is exactly $ -\frac{1}{\sqrt{\lambda}}\log \tau $, where $ \tau $ is the distance from  $ \Sigma $ and $ \lambda $ is the principal eigenvalue of the  MOTS stability operator of $ \Sigma $. In Section \ref{gradient}, we prove the order of the gradient of the solution is $ \tau^{-1} $.
In Section \ref{apriori}, for the  solution $ f_0 $  constructed by J. Metzger in \cite{metzger}, we prove that  the   the coefficients of our estimates for $ f_0 $ depend only on  the initial data set by using the a priori estimates for $ f_0 $. 
We then prove in Section \ref{penrose} that under additional assumptions, we  can apply these results to  the slice  $ \f_0 $  to get a Penrose-like inequality.\par

We start from an initial data set $(\M,g,k)$ with boundary, where $\M$ is a 3-manifold equipped with a Riemannian metric $ g $ together with a symmetric bilinear form $k$ representing the second fundamental form of the time slice $\M$ in space-time. We assume its boundary $ \partial \M $ is consisted of apparent horizons.  An apparent horizon is called a \textit{marginally outer trapped surface} (MOTS) if $ \theta^+(\Sigma) $ vanishes, and  \textit{marginally inner trapped surface} (MITS) if $ \theta^-[\Sigma] $ vanishes. If $ \Sigma $ is a MOTS and there is no other MOTS on the outside of it, we call it an \textit{outermost} MOTS. There might be other apparent horizons inside $ \M $, but we assume that each connected component of $ \partial\M $ has positive distance to the rest of $ \partial\M $ and all other apparent horizons inside $ \M $.

Now we introduce the definition of MOTS stability operator and its principal eigenvalue. For a more detailed investigation  we refer to  \cite{stable}. 

Let $\Sigma\subset \M$ be a MOTS and consider a normal variation of
$\Sigma$ in $\M$, that is a map $F: \Sigma \times (-\epsilon,\epsilon) \to \M$
such that $F(\cdot,0)=id_\Sigma$ and $\dfrac{\partial}{\partial s}\big|_{s=0} F(p,s) =
h\nu$, where $h$ is a function on $\Sigma$ and $\nu$ is the normal of
$\Sigma$. Then the variation of $\theta^+$ is given by
\begin{equation*}
\dfrac{\partial \theta^+[F(\Sigma,s)]}{\partial s}\Big|_{s=0} = L_\Sigma h,
\end{equation*}
where $L_\Sigma$ is a linear elliptic operator of second order along
$\Sigma$, given by
\begin{equation*}
L_\Sigma h =  -^\Sigma\Delta h + 2 S(^\Sigma\nabla h) + h\big(^\Sigma div S -|\chi^+|^2 - |S|^2
+ \tfrac{1}{2}^\Sigma Sc -\mu - J(\nu)\big).
\end{equation*}
In this expression $^\Sigma\nabla$, $^\Sigma div$ and $^\Sigma\Delta$ denote the gradient,
divergence and Laplace-Beltrami operator tangential to $\Sigma$. The
tangential 1-form $S$ is given by  $ S(\nu)=k(\nu,\nu) $ for any $ \nu $ tangential to $ \Sigma $. $\chi^+$ is
the bilinear form $\chi^+ = A + k^\Sigma$, where $A$ is the second
fundamental form of $\Sigma$ in $M$ and $k^\Sigma$ is the projection
of $k$ to $T\Sigma\times T\Sigma$. Furthermore, $^\Sigma Sc$ denotes the scalar curvature of $\Sigma$, $\mu = \dfrac{1}{2}(^\Sigma Sc -|k|^2 + (tr
k)^2)$, and $J = div_M k - d tr k$. It is worth noting that $ L_\Sigma $ is not self-adjoint.\par

However, it was proved in \cite[Lemma 4.1]{stable} that when $ \Sigma $ is compact then there is a real eigenvalue $ \lambda $, called \textit{the principal eigenvalue}, such that the real part of 
any other eigenvalue of $ L_\Sigma $ is greater or equal to $ \lambda $. The corresponding
eigenfunction $ \beta $, $ L\beta =\lambda\beta $ is unique up to a multiplicative constant and can
be chosen to be real and everywhere positive. If $\lambda$ is positive, $\Sigma$ is called
\textit{strictly stable}.  In particular, if $\Sigma$ is strictly stable as a MOTS, there exists an outward deformation strictly increasing $\theta^+$.\par

In this paper, we  use following ways of foliation near the apparent horizon $ \Sigma $. We will use Foliation B in most of our proofs.

\begin{foliationA}
	Denote $ \nu $ the normal vector field of $ \Sigma $ pointing into $ \M $. Defined the foliation $ \Psi_A $ to be the map:
	\begin{equation}
	\Psi_A : \Sigma\times [0,\bar \tau] \to \M:  \mbox{   } (p,\tau)\mapsto \exp_p^{\M}(\tau\nu)	
	\end{equation}
\end{foliationA}

\begin{foliationB}
	Suppose   $ \Sigma $ is a compact boundary component and a strictly stable MOTS with principal eigenvalue  and corresponding eigenfunction $ \lambda>0 $. We further scale so that $ \min_\Sigma\beta=1 $. Defined the foliation $ \Psi_B $ to be the map:
	\begin{align}
	&\Psi_B : \Sigma\times [0,\bar s] \to \M \mbox{ }\mbox{such that}\\ \nonumber
	(1)& \Psi_B(p,0)=p \mbox{ }\mbox{for}\mbox{ }p\in\Sigma\\ \nonumber
	(2)& \dfrac{\partial \Psi_B(p,s)}{\partial s}=\beta(p) \nu_s, \mbox{ }\mbox{where}\mbox{ } \nu_s  \mbox{ }\mbox{is the normal to}\mbox{ }  \Sigma_s:=\Psi(\Sigma,s) \\ \nonumber
	&\mbox{ }\mbox{extending the outward pointing}\mbox{ }
	\mbox{ }\mbox{normal }\mbox{ } \nu_0 \mbox{ }\mbox{on }\mbox{ } \Sigma_0=\Sigma.
	\end{align}
\end{foliationB}

Because  $ \Sigma $ is   compact, $ \beta $ is positive, and   $ \min_\Sigma\beta=1 $,  $ \Psi_A $ and $ \Psi_B $ are comparable.

Our main results are the following:

\begin{theorem}
	\label{thm1}
	Suppose $ (\M,g,k) $ is a smooth initial data set with boundary $ \partial \M=\cup_{i} \Sigma_i $, where each $ \Sigma_i $ is a connected component of the boundary.
	Let $ \Sigma $ be a boundary component which is a compact and strictly stable MOTS  with principal eigenvalue $ \lambda>0 $, and $ f $ be a solution of Eq.\eqref{jang} in an open neighborhood $ V $ of $ \Sigma $, and
	blows up at $ \Sigma $, i.e.  $ f(x)\to +\infty$ when $ x\to \Sigma $. Denote $ \tau(x)$ to be the geodesic distance of a point $ x $ in a neighborhood of $ \Sigma $  to $ \Sigma $. Also denote $ V_c=\{x\in\M|\tau(x)\leq c\}. $\par

	Then there exists $ \tau_0 $ depending only on the local geometry near $ \Sigma $, such that
	$ f(x)+\dfrac{1}{\sqrt{\lambda}} \log\tau(x) $
	is a  bounded  function in $ V_{\tau_0} $.\par
	
	More specifically, under Foliation B, there exist constants  $ a $ and  $ s_0 $ depending only on the local geometry near $ \Sigma $, such that
	the following barrier control for $ f $ holds in $ U_{s_0} $ for any $ 0<s\leq s_1\leq s_0$. Here $ U_{s_0} $ is the region swept out by $ \Sigma_s. $ 
	\begin{equation}\label{thm barrier}
	-\frac{1}{\sqrt{\lambda}}\log\frac{s}{s_1} +\inf_{\Sigma_{s_1}}f +a(s-s_1) \leq
	f(\cdot,s)\leq-\frac{1}{\sqrt{\lambda}}\log\frac{s}{s_1} +\sup_{\Sigma_{s_1}}f-a(s-s_1) 
	\end{equation}
	As a consequence, on each $ \Sigma_s $ for $ s\in(0,s_0) $, if $p_i\in\Sigma,i=1,2 $ are the points such that $ f(p_1,s)=\sup_{\Sigma_s} f $ and  $ f(p_2,s)=\inf_{\Sigma_s} f $, then we have the following gradient estimates for $ f $ at $ p_1, p_2 $:
	\begin{align}
	|\partial_s f(p_1,s)|\leq& \frac{1}{\sqrt{\lambda}s}+a\\
	|\partial_s f(p_2,s)|\geq& \frac{1}{\sqrt{\lambda}s}-a
	\end{align}
	Here $$ \partial_s f(p_i,s):=\lim_{\delta s\to 0}\frac{f(p_i,s+\delta s)-f(p_i,s)}{\delta s}, i=1,2.$$

\end{theorem}
\begin{remark}
	The solution $ f $ might  blow up at some other boundary components, but as long as that boundary component  is a strictly stable MOTS,  the blowup rate of $ f $ near that component can be described by the principal eigenvalue of that component. 
\end{remark}
From the gradient estimate for $ f $ under Foliation B at extreme points on each slice $ \Sigma_s $, we immediately have the following result:
\begin{corollary}\label{corollary1}
	Assume conditions in Theorem \ref{thm1}. Then under Foliation B, if $ f $ is constant on $ \Sigma_{s_1} $ for some $ s_1\in(0,s_0) $, then  the following gradient estimate holds on on  $ \Sigma_{s_1} :$ $$ \frac{1}{\sqrt{\lambda}{s_1}}-a\leq|\partial_s f(\cdot,{s_1})|\leq \frac{1}{\sqrt{\lambda}{s_1}}+a$$
\end{corollary}

Let $ U $ be a neighborhood of $ \Sigma $. We can define a coordinate system on the neighborhood $U\times\R$ of $ \Sigma\times\R $ in $ \M\times\R $ by taking the fourth coordinate $ s $ to be  the parameter in Foliation B.  Let $\bar \Psi : \Sigma \times(-\epsilon,\epsilon)\to \M$ be the map
\begin{equation*}
\bar\Psi
:
\Sigma\times (-\epsilon,\epsilon)\times \R \to \M \times\R
:
(p,s,z) \mapsto \big( \Psi_B(p,s), z \big).    
\end{equation*}

Therefore,  under coordinate  system $ \bar{\Psi} $, following the argument in \cite[Proposition  4]{yau}, the graph of the blowup solution $ f $ near $ \Sigma $ on $ \M $ can be written as the graph of a function $ u $ on the cylindrical end $ \Sigma\times \R $.  In \cite[Section 4]{metzger}, J. Metzger provides a way to extend the $ C^0 $ super  control to $ C^2 $ super  control. Therefore, we are able to get some gradient estimates for   $ f $:

\begin{theorem}
	\label{thm2}
	
	Under the same assumptions as  Theorem \ref{thm1}, suppose $f$ is a blowup solution of Eq.\eqref{jang} in an open neighborhood $ U $ near $ \Sigma $, which is a compact boundary component and a strictly stable MOTS with principal eigenvalue $ \lambda>0 $. Denote $ N=\f $.
	Then,   under coordinate system $ \bar{\Psi} $, \par
	(1) there exist positive constants $\bar z $  and $C_1 $, $C_2 $, $C_3 $, such that $N \cap( U
	\times [\bar z,\infty))$ can be written as the graph of a function
	$u$ over $C_{\bar z}:= \Sigma \times [\bar z,\infty)$,  and
	\begin{gather} 
	|u(p,z)| + |^{C_{\bar z}}\nabla u(p,z)| + |^{C_{\bar z}}\nabla^2 u(p,z)| \leq C_1\exp(-\sqrt{\lambda} z)\\
	|u(p,z)|\geq	C_2\exp(-\sqrt{\lambda} z) \\
	|^{C_{\bar z}}\nabla u(p,z)|\geq	C_3\exp(-\sqrt{\lambda} z) 
	\end{gather} 
	where $^{C_{\bar z}}\nabla$ is the covariant derivative w.r.t. the induced metric on $C_{\bar{z}}$.\par 
	(2) Denote $ \Sigma_s=\Psi_B(\Sigma,s) $. Then there exists constant $ s_0 $, such that the following gradient estimates for $ f $:
	\begin{align}\label{cy barrier3}
	\frac{C_2}{C_1s }\leq	|\partial_s  f (p,s)  |\leq\frac{C_1}{C_3s }\\
	|\nabla^{\Sigma_s} f(p,s)|\leq \frac{\sqrt{2} C_1^2}{C_2 C_3}
	\end{align}
	hold for $ \forall p\in\Sigma$ , $\forall s\in(0, s_0] $. Here $ \nabla^{\Sigma_s}$ denotes the covariant derivative along $ \Sigma_s. $
\end{theorem}
\begin{remark}
	From  Eq. \eqref{thm barrier}, we can see that the barriers for $ f  $ depend  on $ \inf_{\Sigma_{\tau_0}}f $ and $ \sup_{\Sigma_{\tau_0}}f $. This is due to the fact that Jang's equation \eqref{jang}  is invariant under vertical translation. Therefore, $ C_1,C_2, C_3, \bar{ z} $ also depend  on these quantities.	
	If we hope that $ C_1,C_2,C_3,\bar{ z} $ can be determined by the initial data, then   more constraints are needed, e.g. outer boundary condition or decay in infinity.
\end{remark}

If we put more constraints on $ (\M,g,k) $ such that it meets the conditions in \cite[Theorem 3.1, Remark 3.3]{metzger}, then a blowup solution $ f_0 $ of Jang's equation  \eqref{jang}  can be constructed at outermost MOTS, with appropriate a priori estimates.  On this specific Jang's slice $ N_0=\f_0 $, we can get the same estimates as  in Theorem \ref{thm1} and \ref{thm2}, with constants $ C_1,C_2, C_3, \bar{ z} $ only depending on the geometry of the initial data. 

\begin{theorem}\label{thm3}

	Besides the conditions and notations in Theorem \ref{thm2},  we further assume that $ \M $ is a 3-dim asymptotically flat manifold with one end, and  satisfies the dominant energy condition. Also assume that $ \Sigma $ is the only boundary component, and is a compact outermost MOTS. We further assume that there is no MITS in $ \M $. Then a function $ f_0 $ on $ \M $ can be constructed as in \cite[Theorem 3.1]{metzger}, such that $ \J[f_0]=0 $, $ f_0(x)\to 0 $ when $ |x|\to \infty $, and it only blows up at $ \Sigma $. Denote $ N_0=\f_0 $. Then,   under coordinate system $ \bar{\Psi} $, \par
	
	(1) there exist positive constants $\bar z $  and $C_1 $, $C_2 $, $C_3$  only depending on the initial data,  such that $N_0 \cap( U
	\times [\bar z,\infty))$ can be written as the graph of a function
	$u_0$ over $C_{\bar z}:= \Sigma \times [\bar z,\infty)$,  and
	\begin{gather} 
	|u_0(p,z)| + |^{C_{\bar z}}\nabla u_0(p,z)| + |^{C_{\bar z}}\nabla^2 u_0(p,z)| \leq C_1\exp(-\sqrt{\lambda} z)\\ 
	|u_0(p,z)|\geq	C_2\exp(-\sqrt{\lambda} z) \\
	|^{C_{\bar z}}\nabla u_0(p,z)|\geq	C_3\exp(-\sqrt{\lambda} z) 
	\end{gather}
	where $^{C_{\bar z}}\nabla$ is the covariant derivative w.r.t. the induced metric on $C_{\bar{z}}$.\par 
	(2) Denote $ \Sigma_s=\Psi_B(\Sigma,s) $. Then there exists constant $ s_0 $, such that the following gradient estimates for $ f_0 $:
	\begin{align}\label{cy barrier4}
	\frac{C_2}{C_1s } \leq	|\partial_s  f_0 (p,s)  |\leq\frac{C_1}{C_3s }\\
	|\nabla^{\Sigma_s} f_0(p,s)|\leq \frac{\sqrt{2} C_1^2}{C_2 C_3}
	\end{align}
	hold for $ \forall p\in\Sigma$ , $\forall s\in(0, s_0] $. Here $ \nabla^{\Sigma_s}$ denotes the covariant derivative along $ \Sigma_s. $
\end{theorem}
\begin{remark}\label{c}
	Because Foliation A and Foliation B are comparable,   Theorem \ref{thm3} also holds under Foliation A, but with  constants $ C_1', C_2', C_3' $, which are different from  $ C_1, C_2, C_3$  by quantities that can be determined by the local geometry near $ \Sigma $.
\end{remark}
From the construction procedure of $ f_0 $, c.f.\cite{metzger}, we can find that if the initial data set $ (\M,g,k) $ is spherically symmetric, then $ f_0 $ is also  spherically symmetric. Then follows Corollary 3 we have:

\begin{corollary}\label{corollary2}
	Assume conditions and notations in Theorem \ref{thm3} and further assume that  $ (\M,g,k) $ is spherically symmetric, then $ f_0 $ is also  spherically symmetric. Furthermore there exist constants $ a, \tau_0 $ only depending on the local geometry near horizon,  such that 
	$$ \frac{1}{\sqrt{\lambda}\tau}-a\leq|\partial_\tau f_0|\leq \frac{1}{\sqrt{\lambda}\tau}+a$$ holds for $ \tau\in(0,\tau_0] $, where $ \tau $ is the geodesic distance from $ \Sigma $.
\end{corollary}\par

Theorem \ref{thm3} actually provides another way to foliate the neighborhood of $ \Sigma $. The level set of $ f  $: $$ \Sigma_\gamma=\{f_0=-\frac{1}{\sqrt{\lambda}}\log\gamma\}  $$ will form a foliation near $ \Sigma $, because it can be proved from Theorem \ref{thm3} that $ \partial_\tau f $ is uniformly away from zero. Furthermore, from Theorem \ref{thm3} it can be proved that this foliation is comparable with the Foliation A and B. 

\begin{foliationC}
	Under the same assumptions and notations as  Theorem \ref{thm3}. 
	Define the foliation $ \Psi_C $ to be the map:
	\begin{align}
	&\Psi_C : \Sigma\times [0,\bar \gamma] \to \M \mbox{ }\mbox{such that}\\ \nonumber
	(1)& \Psi_C(p,0)=p \mbox{ }\mbox{for}\mbox{ }p\in\Sigma\\ \nonumber
	(2)&  \Psi_C(p,\gamma)=u_0(p,-\frac{1}{\sqrt{\lambda}}\log\gamma)
	\end{align}
	
\end{foliationC}
Then Theorem \ref{thm3} implies the following:
\begin{corollary}
	Denote $ \tau $ to be the geodesic distance to $ \Sigma $, and $ \gamma $ to be the parameter of Foliation C. Then there exist  positive constants $\alpha_1 $, $\alpha_2 $ only depending on the initial data, such that: 
	\begin{align*}
	\alpha_1^{-1}\tau\leq\gamma\leq \alpha_1 \tau\\
	\alpha_2^{-1}\leq\frac{\partial\gamma}{\partial\tau}\leq\alpha_2
	\end{align*}
\end{corollary}

Recall that on a Jang's slice, the   scalar curvature $ \ov{R} $   can be written as:
\begin{equation}\label{syid}
\ov{R}=16 \pi(\mu-J(\omega))+|h-k|^2_{\ov{g}}+2|q|^2_{\ov{g}}-2div_{\ov{g}}(q)
\end{equation}
where $ \ov{g} $ is the metric on Jang's slice, h is the mean curvature of Jang's slice embedded into $ \M\times\R $, and 
\begin{align*}
&\omega_i=\frac{\nabla_i f}{\sqrt{1+|\nabla f|^2}}\\
& q_i=\frac{f^j}{\sqrt{1+|\nabla f|^2}} (h_{ij}-k_{ij})
\end{align*}

From the above theorems,  if  we have one  more constraint  on the coefficients of the gradient estimates in Theorem \ref{thm3}, then 
we are able to prove a Penrose-like inequality with the help of spinor arguments. Furthermore, from \cite[Proposition 2]{yau}, we know that $ |q|_{\ov{g}} $ is bounded near $ \Sigma $. We are also able to prove the same Penrose-like inequality if we put constraint on this upper bound for $ |q|_{\ov{g}} $ :

\begin{theorem}
	\label{thm4}
	Assume the same conditions and notations of Theorem \ref{thm3}. Suppose Condition 1 and 2 to be the following:\par
	
	Condition 1: $$\lambda \frac{C_1^2}{C_3^2}(1+\frac{2 C_1^4}{C_2^2C_3^2})<4, $$ where $ C_1, C_2, C_3 $ are the constants of the gradient estimates in Theorem \ref{thm3};\par
	Condition 2:$$|q|_{\ov{g}}<2\sqrt{\lambda}\mbox{ }\mbox{near}\mbox{ }\Sigma.$$\par
	If one of the above conditions hold, then  we have the following Penrose-like inequalities:
	\begin{equation}
	m\ge \theta\sqrt{\dfrac{|\Sigma|_g}{16\pi}}
	\end{equation}
	where $m$ is the ADM mass, and $ \theta $ is a positive constant given by Eq.\eqref{theta}. $ \theta $  only depends on the geometry of  $ (\M,g,k) $, but does not depend on $ m $. 
\end{theorem}

\section{Introduction  and Review of Existing Results}\label{existing}

In this section we  review the results of Positive Mass Theorem, the existence of solution of Jang's equation in various settings, and the  Riemannian Penrose inequality. \par

The Riemannian Positive Mass Theorem states:
\begin{theorem}
	Let $(M^n,g)$ be a complete asymptotically flat manifold with nonnegative scalar curvature.  If $n<8$ or if $M$ is spin, then the mass of each end is nonnegative.  Moreover, if any of the ends has zero mass, then $(M^n,g)$ is isometric to Euclidean space.
\end{theorem}
The $n<8$ case was first proved by R. Schoen and S.-T. Yau \cite{schyau}, and then by E. Witten for the spin case using a Bochner-Lichnerowicz-Weitzenb\"{o}ck formula \cite{witten} and completed later by T. Parker and C. Taubes \cite{taubes}.\par
It is worth noting that M. Herzlich \cite{herz} proved the Riemannian Positive Mass theorem  for a 3 dimensional manifold with boundary, providing that the mean curvature of the boundary is not too large:
\begin{theorem}
	Let $(M,g)$ be a $ C_\tau^{2,\alpha} $ asymptotically flat manifold of order $ \tau>1/2 $ and scalar curvature in $ L^1 $.    Suppose $ M $ has an inner boundary $ \partial M $, homeomorphic to a 2-sphere, whose mean curvature satisfies $$ H\leq 4\sqrt{\dfrac{\pi}{Area(\partial M)}}. $$ Then, if the scalar curvature of $(M,g)$ is nonnegative, its mass is nonnegative. Moreover, if its mass is zero, then the manifold is flat.
\end{theorem}
The Positive Mass Theorem was then extended to the case of general initial data set  by R. Schoen and S.-T. Yau in \cite{yau} with the help of Jang's equation proposed by P. Jang. 

\begin{theorem}\label{general pmt}
	Let $(\M,g,k) $ be a complete oriented asymptotically flat three dimensional initial data set. Assuming the dominant energy condition, then the ADM mass of each end is non-negative. Moreover, if any of the ends has zero mass, then $(\M, g, k)$ can be isometrically embedded in to four dimensional Minkowski space as a spacelike hypersurface.
\end{theorem} 

This paper also proved the existence of an asymptotically decaying solution of Jang's equation on the asymptotically flat initial data set $ (\M,g,k) $. By asymptotically decaying  we mean  the solution $ f $ satisfies $ \partial^j f=O(|x|^{-j-\frac{1}{2}}) $, $ j=0,1,2,3 $,  on each asymptotically flat end. It was also  proved that the solution $ f $ is well-behaved except that it may blowup at some apparent horizons. 

\begin{theorem}\cite[Proposition 4]{yau}
	\label{original}
	Assume the same conditions in Theorem \ref{general pmt}. There is a sequence $ \{t_i\} $ converging to zero and  open sets $ \Omega_+$, $ \Omega_- $, $ \Omega_0 $, so that if $ f_i $ satisfies $ \J[f_i]=t_i f_i $, we have:
	\begin{enumerate}
		\item The sequence $ \{f_i\} $ converges uniformly to $ +\infty $ (respectively $ -\infty $) on the set $ \Omega_+$ (respectively $ \Omega_-$),   and  $ \{f_i\} $ converges to  a smooth asymptotically decaying solution $ f $ of Eq. \eqref{jang} on $ \Omega_0. $
		\item The sets $ \Omega_+ $ and $ \Omega_- $ have compact closure, and $ \M=\bar{\Omega}_+\cup\bar{\Omega}_-\cup\bar{\Omega}_0 $. Each boundary component $ \Sigma $ of $ \Omega_+ $ (respectively $ \Omega_- $) is a smooth embedded two-sphere satisfying $ H_\Sigma-\text{tr}_\Sigma(k_{ij})=0 $ (respectively $ H_\Sigma+\text{tr}_\Sigma(k_{ij})=0 $). Moreover, no two connected components of $ \Omega_+ $ can share a common boundary.
		\item The graphs $ N_i$ of $ f_i $ converge smoothly to   a   properly embedded submanifold $ N $ of $ \M \times\R $. Each connected component of $ N $
		is either a component of the graph $ f$, or the cylinder $ \Sigma\times\R\subset \M\times\R $ over a boundary component $ \Sigma $ of $ \Omega_+ $ or $ \Omega_- $. Any two connected components of $ N_0 $ are separated by a positive distance.
	\end{enumerate}
	
\end{theorem}
It was also proved in \cite{yau} that  if this blowup happens,  then $N $ can be written as the graph of a function  $ u $ on the corresponding cylindrical end, and the derivatives of $ u $ up to second order also tend to zero asymptotically on this cylindrical end, c.f.\cite[Corollary 2]{yau}:

\begin{theorem}\cite[Corollary 2]{yau}
	\label{thm:sideways}
	Assume the conditions of Theorem \ref{original} and let $\Sigma$ be a connected component
	of the apparent horizons, on which $ f $ tends to $ +\infty $ ($ -\infty $ respectively). Let $U$ be a neighborhood of $\Sigma$ with positive
	distance to any other apparent horizons in $ \M $.
	
	Then for all $\epsilon>0$ there exists $\bar z = \bar z(\epsilon)$,
	depending also on the geometry of $(\M,g,k)$, such that $N \cap (U
	\times [\bar z,\infty))$ can be written as the graph of a function
	$u$ over $C_{\bar z}:= \Sigma \times [\bar z,\infty)$, so that
	\begin{equation*}
	|u(p,z)| + |^{C_{\bar z}}\nabla u(p,z) | + |^{C_{\bar z}}\nabla^2 u(p,z)| < \epsilon.
	\end{equation*}
	for all $(p,z)\in C_{\bar z}$. Here, $^{C_{\bar z}}\nabla$ denotes
	covariant differentiation along $C_{\bar z}$.
\end{theorem}

To prove the Penrose inequality, it is necessary to capture the information of the apparent horizons. Therefore, one needs to prescribe the boundary condition for Jang's equation at horizons. One approach is to let the solution blow up  near $\Sigma$  and approximate a cylinder over $\partial\M$. In \cite{metzger}, J. Metzger showed that if $\Sigma$ is an outermost MOTS, then in fact there \textit{must} exist a solution to \eqref{jang} which blows up at $\Sigma$ (and only at $\Sigma$, provided there are no MITSs in $ \M $).

\begin{theorem}\cite[Theorem 3.1]{metzger}
	\label{thm:blowup}
	If $(\M,g,k)$ be an initial data set with $\partial \M = \partial^- \M \cup
	\partial^+ \M$ such that $\partial^-\M$ is an outermost MOTS,
	$\theta^+[\partial^+\M]>0$ and $\theta^-[\partial^+\M]<0$, then there exists
	an open set $\Omega_0\subset \M$ and a function $f_0: \Omega_0\to \R$
	such that
	\begin{enumerate}
		\item $\M\setminus\Omega_0$ does not intersect $\partial \M$,
		\item $\theta^-[\partial \Omega_0] =0$ with respect to the normal vector
		pointing into $\Omega_0$, 
		\item $\J[f_0] = 0$,
		\item $N^+ = \text{graph} f_0 \cap \M\times \R^+$ is asymptotic to the cylinder $\partial ^- M
		\times\R^+$,
		\item $N^- = \text{graph} f_0 \cap \M\times \R^-$ is asymptotic to the cylinder $\partial \Omega_0
		\times\R^-$, and
		\item $f_0 |_{\partial ^+ \M} = 0$.
	\end{enumerate}
\end{theorem}

This existence theorem still holds if the boundary condition is changed from  $f_0 |_{\partial ^+ \M} = 0$ to asymptotic decaying if the initial data set is asymptotically flat (c.f. \cite[Remark 3.3]{metzger}).  \cite[Proposition 3.1]{jek} provided a proof for this case for higher dimension $ 3\le n\le 7 $. Moreover, if there is no MITS in $ \M $, then $ \Omega_0=\M $.\par

With the  blowup solution $ f_0 $ constructed in Theorem \ref{thm:blowup} , J. Metzger showed
that under the assumption of strict stability, the graph of $ f_0 $ can be written as the graph of a function $ u_0 $ on the cylindrical end, whose decay rate  is
exponential with a power directly related to the principal eigenvalue
of the MOTS. The general idea is to show the existence of a super-solution
with at most logarithmic blowup of the desired rate, c.f.\cite[Theorem 4.2, Theorem 4.4]{metzger}: 
\begin{theorem}\cite[Theorem 4.2]{metzger}
	\label{thm:exponential}
	Let $N_0 = \f_0$ be the manifold constructed in Theorem \ref{thm:blowup} and 
	assume the situation of Theorem~\ref{thm:sideways}. Then there exists $\bar z = \bar z(\epsilon)$,
	depending also on the geometry of $(\M,g,k)$, such that $N_0 \cap (U
	\times [\bar z,\infty))$ can be written as the graph of a function
	$u_0$ over $C_{\bar z}:= \Sigma \times [\bar z,\infty).$
	If in addition $\Sigma$ is
	strictly stable with principal eigenvalue $\lambda >0$, then
	for all $\delta < \sqrt{\lambda}$ there exists $c=c(\delta)$ depending only
	on the data $(\M,g,k)$ and $\delta$ such that
	\begin{equation*}
	|u_0(p,z)| + |^{C_{\bar z}}\nabla u_0(p,z)| + |^{C_{\bar z}}\nabla^2 u_0(p,z)| \leq c\exp(-\delta z).
	\end{equation*}
	where  $^{C_{\bar z}}\nabla$ is the covariant derivative w.r.t. the induced metric on $C_{\bar{z}}$. 
\end{theorem}
A non-existence theorems for the blowup solution with higher blowup rate was also proved in \cite{metzger}:
\begin{theorem}\cite[Theorem 4.4]{metzger}
	\label{thm:rigidity}
	Under the assumptions of Theorem~\ref{thm:exponential} there are no
	solutions $h : \Sigma\times [0,\infty) \to \R$ to the equation \eqref{jang}
	with decay 
	\begin{equation*}
	|h(p,z)| + |^{C_{\bar z}}\nabla h(p,z)| + |^{C_{\bar z}}\nabla^2 h(p,z)|
	\leq
	C \exp(-\delta z)
	\end{equation*}
	such that $\delta > \sqrt{\lambda}$ and $h>0$. Here, $^{C_{\bar z}}\nabla$ is the covariant derivative w.r.t. the induced metric on $C_{\bar{z}}$. 
\end{theorem}

\begin{remark}
	It is worth noting that  Theorem \ref{thm:rigidity}  can be implied by our theorem.
\end{remark}

In \cite{coupled}, H.  Bray and M.   Khuri proposed the generalized Jang's equation, in an attempt to prove the Penrose inequality in the setting of general initial data. It is defined as 
\begin{equation}\label{7}
\left(g^{ij}-\frac{\phi^{2}f^{i}f^{j}}{1+\phi^{2}|\nabla f|^{2}}\right)
\left(\frac{\phi\nabla_{ij}f+\phi_{i}f_{j}+\phi_{j}f_{i}}
{\sqrt{1+\phi^{2}|\nabla f|^{2}}}-k_{ij}\right)=0.
\end{equation}

An appropriate choice of $ \phi $ will provide a proof for Penrose inequality. However, this  involves solving coupled equations, and the full existence theorem for the solution remains unsolved, except for spherical symmetric case,  c.f. \cite{coupled}.
In \cite{han}, Q. Han and M. Khuri proved the existence of blowup solution of generalized Jang's equation, in the case when $ \phi $ is fixed and independent of the solution, on an asymptotically flat initial data set. They also gave the blowup rate estimates for the solution they constructed. 

Since this article is not focusing on the generalized version Eq.\eqref{7}, we only include their theorem for the case of $ \phi \equiv1$:

\begin{theorem}\cite[Theorem 1.1]{han}\label{han}(case  $ \phi \equiv1$)
	Suppose that $(\M,g,k)$ is a smooth, asymptotically flat initial data
	set, with outermost apparent horizon boundary $\partial \M$ composed of MOTS  $ \partial^+\M $ and MITS  $ \partial^-\M $. Denote by $\tau$ the distance function from $ \partial \M $ and  $\Sigma_{\tau}$  the level sets of the geodesic flow emanating from $\partial \M$. Assume $ c^{-1}\tau\leq\theta^{\pm}(\Sigma_{\tau})\leq c\tau $ for some constants   $c>0$, then
	there exists a smooth asymptotically decaying solution $f$ of the 
	Jang's equation \eqref{jang}, such that
	$f(x)\rightarrow\pm\infty$ as $x\rightarrow\partial^{\pm}\M$. More
	precisely, in a neighborhood of $\partial^{\pm}\M$:
	\begin{equation}
	-\alpha^{-1}\log\tau+\beta^{-1} \leq \pm f
	\leq -\alpha\log\tau+\beta
	\end{equation}
	
	for some positive constants $\alpha$ and $\beta$.

\end{theorem}

\begin{remark}
	The lower bound blowup control in the above theorem is not only valid for the solution constructed in \cite{han}, but also valid for all blowup solutions. This is because the family of sub barriers used to prove the lower bound in Theorem \ref{han} are finite at horizon, thus  can be applied for all possible blowup solutions. However in this article,  a different family of sub barriers are constructed, which  provide a better estimate for  the blowup solution. 
\end{remark}

In \cite{williams}, C. Williams showed that for  large classes of spherically symmetric initial data, there are solutions of the Jang equation which blow up at non-outermost MOTSs, i.e. MOTSs which lie strictly inside of other MOTSs, and even inside of strictly outer trapped surfaces. \par

Until now there is no result about the uniqueness of  the solution of Jang's equation on a general initial data set. For time symmetric case, Jang's equation  is reduced to a Jenkins-Serrin type equation. Under some boundary  constrain,  M. Eichmair and  J. Metzger were able to prove the uniqueness for solution of Jang's equation in time symmetric case, c.f. \cite[Section 5]{jek}.

Thus in the time symmetric case, the existence and uniqueness of the asymptotically decaying solution of Jang's equation \eqref{jang} which blows up  at outermost MOTS has been proved. We also know by \cite[Theorem 1.1]{han} that in this case the blowup rate near horizon is $ \log\tau $, where $ \tau  $ is the distance from the horizon. However, we still don't know the coefficient of $ \log\tau $  in this case, and whether there are any other lower order blowup terms. In this paper, we will show that even in the non-time-symmetric case, for a blowup solution $ f $ of Jang's equation \eqref{jang} at a strictly stable MOTS with principal eigenvalue $ \lambda $,  the only blowup term of $ f $ is exactly $  -\dfrac{1}{\sqrt{\lambda}} \log\tau$. We are also able to prove that the gradient of $ f $ is of order $ \tau^{-1} $. For these two parts of the result we don't need extra assumptions on spacetime such as the dominant energy condition or  asymptotic flatness. This is because our estimation is local near horizon.\par

The third part of our result is to apply these estimates to the Jang's slice $ N_0=\f_0 $ constructed by J. Metzger in \cite{metzger} to prove a Penrose-like inequality. \par 
The  Penrose inequality may be thought of as a refinement of the Positive Mass Theorem when black holes are present. It relates the total ADM mass (of a chosen
end) $m$ to the area $A$ of its outermost minimal area enclosure. The inequality states that: 
\begin{equation}
m\geq\sqrt{\frac{A}{16\pi}}.
\end{equation}
And furthermore it asserts that if equality holds and the outermost minimal area enclosure is the boundary of an open bounded domain $U\subset M$, then $(M-U,g)$ admits an isometric embedding into
the Schwarzschild spacetime with second fundamental form given by $k=0$.

The Riemannian Penrose inequality (time symmetric case) states:

\begin{theorem}[Riemannian Penrose Inequality]
	Let $(M^n,g)$ be a complete asymptotically flat manifold with nonnegative scalar curvature, where $n<8$.  Fix one end. Let $m$ be the mass of that end, and let $A$ be the area of an outer minimizing horizon (with one or more components). Let $\omega_{n-1}$ be the area of the standard unit $(n-1)$-sphere.  Then
	$$m\geq {1\over2} \left({A\over \omega_{n-1}}\right)^{n-2\over n-1},$$
	with equality if and only if the part of $(M,g)$ outside the horizon is isometric to a Riemannian Schwarzschild manifold outside its unique outer minimizing horizon.
\end{theorem}
The $n=3$ case was  proved by G. Huisken and T. Ilmanen using inverse mean curvature flow method (c.f. \cite{huisken}), and also by H. Bray  using a conformal flow of metrics (c.f.  \cite{bray}). Then H. Bray and D. Lee generalized the result  to the $ n<8 $ case (c.f.   \cite{lee}). M. Herzlich also gave a Penrose-like inequality in  his paper \cite{herz}. \par

For the case of general initial data set $ (\M,g,k) $ when the dominant energy condition $ \mu\geq |J| $ is satisfied, it is possible to obtain a non-sharp Penrose-like inequality with the help of Jang's equation. Recall that on a Jang's slice, the   scalar curvature $ \ov{R} $   can be written as:
\begin{equation}\label{syid}
\ov{R}=16 \pi(\mu-J(\omega))+|h-k|^2_{\ov{g}}+2|q|^2_{\ov{g}}-2div_{\ov{g}}(q)
\end{equation}
where $ \ov{g} $ is the metric on Jang's slice, h is the mean curvature of Jang's slice embedded into $ \M\times\R $, and 
\begin{align*}
&\omega_i=\frac{\nabla_i f}{\sqrt{1+|\nabla f|^2}}\\
& q_i=\frac{f^j}{\sqrt{1+|\nabla f|^2}} (h_{ij}-k_{ij})
\end{align*}

For the case when vector field $ q $ vanishes at a cylindrical end of Jang's slice, M.  Khuri proved  in \cite{Khuri} a Penrose-like inequality for an initial data set with charge. However, until now the Penrose inequality  for the case when $ q $ does not vanish at the cylindrical end of Jang's slice  still remains open. For example,   in Schwarzschild spacetime, we can construct a blowup solution for Jang's equation on its time symmetric slice, and we can calculate that the vector field $ q $ does not  vanish at the cylindrical end of this Jang's slice. In this article,  we will prove a Penrose-like inequality for this case, based on a sharp estimate of the cylindrical end of  Jang's slice.

\section{Barriers Construction}\label{barrier}

Because Theorem \ref{thm1} is a local estimate of the solution   at each strictly stable MOTS boundary component,  we can do the estimate  separately. It is equivalent to prove our result on one of these components. We denote this compact and strictly stable component as $ \Sigma $ with principal eigenvalue $ \lambda>0 $. For convenience we assume  $ f(x)\to +\infty$ when $ x\to \Sigma $.\par

All of our computation in Section \ref{barrier} will be under Foliation B. Under this foliation we can choose $\bar s>0$ small enough such that the
surfaces $\Sigma_s=\Psi_B(\Sigma,s)$ with $s\in[0,\bar s]$ form a
local foliation near $\Sigma$ with lapse $\beta$ such that
\begin{equation*}
\lambda\beta s- \Lambda s^2\leq	\theta^+[\Sigma_s] \leq \lambda\beta s+ \Lambda s^2.
\end{equation*}
for some constant $ \Lambda $.\par
Denote the region swept out by these $\Sigma_s$ by $U_{\bar s}$. Note that
$\partial U_{\bar s} = \Sigma \cup \Sigma_{\bar s}$ and $dist(\Sigma_{\bar
	s},\Sigma)\geq\bar s$. We can assume that $dist(\Sigma_{\bar s},\partial
\M)>0$. On $U_{\bar s}$ we consider a test function of the form $v=\phi(s)$. For
such functions Jang's operator can be computed as follows, cf.~\cite{metzger}, ~\cite{Andersson}. Note that $ \dfrac{\partial \beta}{\partial s} $ does not appear because our $ \beta $ is chosen to be a function on $ \Sigma $.
\begin{align}
\label{local}
\J[\phi]
=&
\frac{\phi'}{\beta\sqrt{1 + \beta^{-2}(\phi')^2}}\theta^+
- \left( 1+ \frac{\phi'}{\beta\sqrt{1 + \beta^{-2}(\phi')^2}}\right) P\\\nonumber
&- \frac{ k(\nu,\nu)}{1 + \beta^{-2}(\phi')^2}
+ \frac{\phi''}{\beta^2(1 + \beta^{-2}(\phi')^2)^{\frac{3}{2}}},
\end{align}
where  $\phi'$ denotes the
derivative of $v=\phi(s)$ with respect to $s$. The quantities $\theta^+$,
$k(\nu,\nu)$ and $P=tr k-k(\nu,\nu)$ are computed on the respective $\Sigma_s$.

\subsection{Super Estimate: Order Control}

The most important part of our proof in this section is the order control of the solution $ f $. If  we know nothing about the blowup order of $ f $ beforehand, it will be  impossible to apply the comparison principle in $ U_{\ov{s}} $, even if we have   a super solution.

In this section we will prove that for any blowup solution $ f $ of Eq.\eqref{jang}, there exist  constants $ c_1>0,  s_1>0$ which only depend on the local  geometry of the initial data set near the horizon, such that $ f<-c_1 \log s $ holds for $ s\in(0,s_1] $. In order to prove this inequality, in Proposition \ref{w} below, we construct a piecewise smooth function $ W(s) $ iteratively, and prove that $ f<W $ by induction. Moreover, we observe for $ s\to 0 $ that the blowup order of $ W(s) $ is exactly $ -c_1\log s $ for some positive constant $ c_1 $, thus  the inequality gets proved. In the  next section, we will refine $ c_1 $ to $ \dfrac{1}{\sqrt{\lambda}} $ by  another family of super barriers which will be constructed in Lemma \ref{superlemma2}.

We first construct a family of auxiliary super barriers under the Foliation B. We will use these auxiliary super barriers to construct the upper bound function $ W(s) $ in Proposition \ref{w}.

For $ \epsilon>0 $, we set functions $w_\epsilon $ to be the following:
\begin{lemma}
	\label{superlemma1}
	For any $ \epsilon>0 $, we construct the following functions $ w_\epsilon $ on $ s\in(\epsilon,\bar s) $:
	\begin{equation}
	w_\epsilon (s)=-\log(s-\epsilon)
	\end{equation}

	These functions have the following properties:
	\begin{enumerate}
		\item $ \lim_{s\to\epsilon}w_\epsilon (s)=+\infty $
		\item There exist constants $ \epsilon_0>0 $ and $ \alpha>0 $ which  only depend on $ (\M,g,k) $, s.t. $ \forall 0<\epsilon\le\epsilon_0 $, $ w_\epsilon (s) $ is a  super solution on $ s\in(\epsilon,(1+\alpha)\epsilon] $.

	\end{enumerate}

\end{lemma}
\begin{proof}
	Property (1) is obvious, thus we only prove Property (2) here.
	If we  plug   $\phi=w_\epsilon $  into Eq.\eqref{local}, then for $ s\in(\epsilon,(1+\alpha)\epsilon] $, ($ \alpha $ is a constant to be determined),   each term  of Eq.\eqref{local} will be:
	\begin{align*}
	\frac{w_\epsilon'}{\beta\sqrt{1 + \beta^{-2}(w_\epsilon')^2}}=&-\frac{1}{\sqrt{1+\beta^2 (s-\epsilon)^2}}=-1+O((s-\epsilon)^2)\\
	\theta^+=&\lambda\beta s +O(s^2)\\
	\frac{1}{ 1 + \beta^{-2}(w_\epsilon')^2 }=&\frac{\beta^2 (s-\epsilon)^2}{1+\beta^2 (s-\epsilon)^2}=\beta^2(s-\epsilon)^2+O((s-\epsilon)^4)\\
	\frac{w_\epsilon''}{\beta^2( 1 + \beta^{-2}(w_\epsilon')^2)^{\frac{3}{2}}}=&\frac{\beta(s-\epsilon)}{(1+\beta^2(s-\epsilon)^2)^{\frac{3}{2}}}=\beta(s-\epsilon)+O((s-\epsilon)^3)
	\end{align*}
	
	Thus 
	\begin{align*}
	\J[w_\epsilon]=&-\lambda\beta s +\beta(s-\epsilon)+O((s-\epsilon)^2)\\
	=&(1-\lambda)\beta s-\beta\epsilon+O((s-\epsilon)^2)\\
	=&(1-\lambda)\beta s-\beta\epsilon+O( \epsilon^2)
	\end{align*}
	for $ s\in(\epsilon,(1+\alpha)\epsilon] $.\\
	Thus if we set $ \alpha $ to be a constant s.t. $(1-\lambda)(1+ \alpha)<1-\dfrac{\lambda}{2} $,   we will have:
	\begin{equation*}
	\J[w_\epsilon]\leq(1-\frac{\lambda}{2})\beta\epsilon-\beta\epsilon +O(\epsilon^2)=-\frac{\lambda}{2}\beta\epsilon+O(\epsilon^2)<0
	\end{equation*}
	holds for $ \epsilon\in(0,\epsilon_0] $, where $ \epsilon_0 $ is a constant only related to the geometry of $ (\M,g,k) $ near $ \Sigma $.	
\end{proof}

Now we focus on a  blowup solution $ f $ satisfying   Eq.\eqref{jang} on $ U_{\bar s} $. For simplicity we assume that $ \bar{s}>(1+\alpha)\epsilon_0 $.

By the comparison principle, c.f. \cite[Section 10]{Gilbarg}, we get the following result: 
\begin{proposition}
	There exist constants $ \epsilon_0 $ and $ \alpha $, s.t.
	for any function $ f $ that satisfies  \eqref{jang}  on $U_{ \bar s} $ and  $ f\to+\infty $ when $ s\to 0 $, the following inequality :
	\begin{equation}
	\label{superep}
	f(\cdot,s) \le \sup_{q\in\Sigma}f(q,(1+\alpha)\epsilon)+\log(\frac{\alpha\epsilon}{s-\epsilon})
	\end{equation}
	holds for any $\epsilon\in(0,\epsilon_0]$ and $ s\in(\epsilon,(1+\alpha)\epsilon] $. 
\end{proposition}

\begin{proof}
	
	We assume $ \epsilon_0 $, $ \alpha$  the same as in Lemma \ref{superlemma1} and $ (1+\alpha)\epsilon_0 < \bar s $. Then by   Property (2) of Lemma \ref{superlemma1}, we know that for $ \forall \epsilon\in(0,\epsilon_0) $, $ \ov{ w}_\epsilon(s)=w_\epsilon(s)- w_\epsilon((1+\alpha)\epsilon)+\sup_{q\in\Sigma}f(q,(1+\alpha)\epsilon) $ is a super solution on $ s\in(\epsilon,(1+\alpha)\epsilon] $. This super solution satisfies $ \ov{ w}_\epsilon((1+\alpha)\epsilon)\geq f(\cdot,(1+\alpha)\epsilon) $. Moreover,  by Property (1) in Lemma \ref{superlemma1}, and the fact that $f(\cdot,\epsilon) $ is finite, we know by the comparison principle that this super solution is actually a super barrier for $ f $ for $ s\in(\epsilon,(1+\alpha)\epsilon] $.  	
\end{proof}

In the above proposition we find  super barriers which are only effective in a region which becomes infinitely small when $ \epsilon\to 0 $. This is not enough for obtaining a upper bound for the blowup rate. In  Proposition \ref{w} below we use the above $ w_\epsilon $ to iteratively  construct a upper bound function $ W(s) $ for $ f $ on a fixed region $ s\in(0,(1+\alpha)\epsilon_0] $.

For any function $ f $ that satisfies  \eqref{jang}  on $U_{ \bar{s}} $ and  $ f\to+\infty $ when $ s\to 0 $, we construct the following upper bound $ W(s) $   for $ f $ iteratively for $ s\in(0,(1+\alpha)\epsilon] $:

\begin{proposition}\label{w}
	For a constant $ d\in(0,1)$,  and  function $ f $ that satisfies  Eq.\eqref{jang}  on $U_{ \bar{s}} $ and  $ f\to+\infty $ when $ s\to 0 $, we construct $ W(s) $ iteratively in the following way:
	\begin{enumerate}
		\item For $ s\in[(1+d \alpha)\epsilon_0,(1+\alpha)\epsilon_0], $ define
		$$	W(s)=\log(\frac{\alpha\epsilon_0}{s-\epsilon_0})+\sup_{q\in\Sigma}f(q,(1+\alpha)\epsilon_0)
		$$
		\item For $ n\ge 0 $, if $ W(s) $  is defined in $ s\in\big[\dfrac{(1+d\alpha)^{n+1}}{(1+\alpha)^n} \epsilon_0,\dfrac{(1+d\alpha)^n}{(1+\alpha)^{n-1}}\epsilon_0\big) $, then define
		\begin{equation*}
		W(s)=\log\left(\frac{\alpha\dfrac{(1+d\alpha)^{n+1}}{(1+\alpha)^{n+1}}\epsilon_0}{s-\dfrac{(1+d\alpha)^{n+1}}{(1+\alpha)^{n+1}}\epsilon_0}\right)+W\big(\dfrac{(1+d\alpha)^{n+1}}{(1+\alpha)^n} \epsilon_0\big)
		\end{equation*} 
		for $ s\in\big[\dfrac{(1+d\alpha)^{n+2}}{(1+\alpha)^{n+1}} \epsilon_0,\dfrac{(1+d\alpha)^{n+1}}{(1+\alpha)^{n}}\epsilon_0\big)  $

	\end{enumerate}
	Then $ W(s) $  will be a piecewise smooth function defined on $ s\in(0,(1+\alpha)\epsilon_0] $. Furthermore,  $ W $ will be an upper bound for $ f $, and the following inequality 
	\begin{equation}
	\label{superfix}
	f(p,s)\le W(s)
	\end{equation}
	holds for $ \forall s\in(0,(1+\alpha)\epsilon_0] $ and $ \forall p\in\Sigma $.
	
\end{proposition}
\begin{proof}
	We will prove Ineq.\eqref{superfix}  on $ s\in(0,(1+\alpha)\epsilon_0] $ by induction.

	First, plug  $ \epsilon=\epsilon_0 $  into Ineq.\eqref{superep}. By the fact that $ (1+d \alpha)\epsilon_0>\epsilon_0 $, we get Ineq.\eqref{superfix} holds on $ s\in[(1+d\alpha)\epsilon_0 ,(1+\alpha)\epsilon_0 ) $.

	Next, if Ineq.\eqref{superfix} holds for $ s\in\big[\dfrac{(1+d\alpha)^{n+1}}{(1+\alpha)^n} \epsilon_0,\dfrac{(1+d\alpha)^n}{(1+\alpha)^{n-1}}\epsilon_0\big) $ for some $ n\ge 0 $, (especially we have $ f(\cdot,\dfrac{(1+d\alpha)^{n+1}}{(1+\alpha)^n} \epsilon_0)\le W(\dfrac{(1+d\alpha)^{n+1}}{(1+\alpha)^n} \epsilon_0) $), we apply Ineq.\eqref{superep} for $ \epsilon=\dfrac{(1+d\alpha)^{n+1}}{(1+\alpha)^{n+1}} \epsilon_0 $ and by the fact that $ \dfrac{(1+d\alpha)^{n+2}}{(1+\alpha)^{n+1}} \epsilon_0 >\epsilon=\dfrac{(1+d\alpha)^{n+1}}{(1+\alpha)^{n+1}} \epsilon_0$, we have:
	\begin{align*}
	f(p,s)\le&\sup_{q\in\Sigma} f(q,\frac{(1+d\alpha)^{n+1}}{(1+\alpha)^{n}} \epsilon_0) + \log\left(\frac{\alpha\dfrac{(1+d\alpha)^{n+1}}{(1+\alpha)^{n+1}}\epsilon_0}{s-\dfrac{(1+d\alpha)^{n+1}}{(1+\alpha)^{n+1}}\epsilon_0}\right)\\
	\le& W\big(\dfrac{(1+d\alpha)^{n+1}}{(1+\alpha)^n} \epsilon_0\big)+\log\left(\frac{\alpha\dfrac{(1+d\alpha)^{n+1}}{(1+\alpha)^{n+1}}\epsilon_0}{s-\dfrac{(1+d\alpha)^{n+1}}{(1+\alpha)^{n+1}}\epsilon_0}\right)\\
	=&W(s)
	\end{align*}
	holds 	for $ s\in\big[\dfrac{(1+d\alpha)^{n+2}}{(1+\alpha)^{n+1}} \epsilon_0,\dfrac{(1+d\alpha)^{n+1}}{(1+\alpha)^{n}}\epsilon_0\big) $, $ p\in\Sigma $.
	
	Thus by induction we know that Ineq.\eqref{superfix} holds for $ s\in\big[\dfrac{(1+d\alpha)^{n+1}}{(1+\alpha)^n} \epsilon_0,\dfrac{(1+d\alpha)^n}{(1+\alpha)^{n-1}}\epsilon_0\big) $, for all interger $ n\ge 0 $, thus it holds for $ s\in(0,(1+\alpha)\epsilon_0] $.
\end{proof} 
In the above construction for $ W $, we can also prove  inductively that $ W\big(\dfrac{(1+d\alpha)^{n}}{(1+\alpha)^{n-1}} \epsilon_0\big)= -n\log d +\sup_{q\in\Sigma}f(q,(1+\alpha)\epsilon_0)$ for any non-negative interger $ n $.

Thus  we know that for  for  $ s\in\big[\dfrac{(1+d\alpha)^{N+1}}{(1+\alpha)^N} \epsilon_0,\dfrac{(1+d\alpha)^N}{(1+\alpha)^{N-1}}\epsilon_0\big) $ :
\begin{align*}
f(\cdot,s)\le& W(s)\\
\le& W\big( \dfrac{(1+d\alpha)^{N+1}}{(1+\alpha)^N} \epsilon_0 \big) \\
=&-(N+1)\log d+\sup_{q\in\Sigma}f(q,(1+\alpha)\epsilon_0) \\
=& (N+1)\log \frac{1}{d}+\sup_{q\in\Sigma}f(q,(1+\alpha)\epsilon_0) \\
<& (1-\log_{\frac{1+ \alpha}{1+d\alpha}}(\frac{s}{(1+\alpha)\epsilon_0}))\log \frac{1}{d}+\sup_{q\in\Sigma}f(q,(1+\alpha)\epsilon_0)
\end{align*}
By setting $ d=\dfrac{1}{2} $, we get:
\begin{align*}\label{half}
f(\cdot,s)\le&
-\frac{\log 2}{\log(1+\alpha)-\log(1+\frac{\alpha}{2})}\log s +\sup_{q\in\Sigma}f(q,(1+\alpha)\epsilon_0)\\
&+\log 2+\frac{\log 2           \log((1+\alpha)\epsilon_0)}{\log(1+\alpha)-\log(1+\frac{\alpha}{2})}
\end{align*}

\begin{remark}
	By the fact that
	\begin{equation*}
	\lim_{d\to1^-}\frac{\log d}{\log \frac{1+d \alpha}{1+\alpha} }=\frac{\alpha+1}{\alpha}
	\end{equation*}
	we know  the opitimal value of the coefficient of $ -\log s $ in the Ineq.\eqref{half} will be $ \frac{\alpha+1}{\alpha} $, and the corresponding inequality will be:
	\begin{equation*}
	f(\cdot,s)\le-\frac{\alpha+1}{\alpha}\log\frac{s}{(1+\alpha)\epsilon_0}+\sup_{q\in\Sigma}f(q,(1+\alpha)\epsilon_0) 
	\end{equation*}
\end{remark}

Because $ \epsilon_0,\alpha $ are constants determined by $ (\M,g,k) $,  thus for any blowup solution $ f $, we have the following upper bound for the blowup order:
\begin{proposition}
	\label{1}
	There exist constants  $ s_1 $, $ c_1 $, and $ C_1 $ only depend on the local geometry of $ (\M,g,k) $ near $ \Sigma $, s.t. 
	for any function $f $ that satisfies Eq.\eqref{jang}  on $U_{ \bar s} $ and  $ f\to+\infty $ when $ s\to 0 $, the following inequality :
	\begin{equation}
	\label{super1}
	f(p,s) \le \sup_{q\in\Sigma}f(q,s_1)-c_1\log s+C_1
	\end{equation}
	holds for   $\forall s\in(0,s_1] $ and $ p\in \Sigma $. 
\end{proposition}

\subsection{Super Estimate: Coefficient Refinement}
In the previous section we prove that any blowup solution of Jang's equation in our case cannot blow up faster than $ -c_1\log s$ for some positive constant $ c_1 $ which only depends on $ (\M,g,K) $ (c.f. Ineq.\eqref{super1}). In this section we are going to refine this result to $-\dfrac{1}{\sqrt{\lambda}}\log s.$

To achieve this, we need to construct a new family of super barriers. We first
compute the expansion of $ \J[v_{a,\gamma}] $ for some test functions near the  horizon.\par

\begin{lemma}
	\label{taylor}
	For $ \gamma>0 $ and constant $ a $, we construct the following functions $ v_{a,\gamma} $ :
	\begin{equation*}
	v_{a,\gamma}(s)=\dfrac{1}{\sqrt{\lambda}}\int_{s}^{\bar s} \frac{1}{x^{\gamma}}dx+a s
	\end{equation*}	
	for $ s\in(0,\bar s] $.
	
	These functions have the following properties:
	\begin{enumerate}
		
		\item For $ \gamma>1 $, $v_{a,\gamma}(s) $ blowup at the rate of $ \dfrac{s^{1-\gamma}}{\sqrt{\lambda}(\gamma-1)} $ when $ s $ goes to $ 0 $ . For $ \gamma<1 $, $v_{a,\gamma}(s) $ is bounded in $ U_{\bar{ s}} $.
		\item $ v'_{a,\gamma}(s)= -\dfrac{1}{\sqrt{\lambda}s^{\gamma}}+a$, 	  $ v''_{a,\gamma}(s)= \dfrac{\gamma}{\sqrt{\lambda}s^{\gamma+1}}$
		\item Denote $ \sigma^2=(a^2+\beta^2)s^{2\gamma}-\dfrac{2a}{\sqrt{\lambda}} s^\gamma+\dfrac{1}{\lambda} $. Then for $ \gamma $ close to 1, $\sigma^3 \J[v_{a,\gamma}] $ has the following expansion in $ U_{\ov{s}} $:
		\begin{align}\label{expansion}
		\sigma^3\J[v_{a,\gamma}]=&-\frac{\beta}{\sqrt{\lambda}} s+\frac{\beta\gamma}{\sqrt{\lambda}}s^{2\gamma-1}\\ \nonumber
		& -\frac{1}{\lambda\sqrt{\lambda}} (\theta^+-\lambda\beta s)
		+3 a \beta s^{\gamma+1}
		-\frac{\beta^2}{\sqrt{\lambda}}(k(\nu,\nu)+\frac{P}{2})s^{2\gamma}\\ \nonumber
		&  +O(s^{3\gamma}+s^{2\gamma+1}+s^{\gamma+2}+s^3)	
		\end{align}
	\end{enumerate}
\end{lemma}	
\begin{proof}
	Properties (1)(2) are obvious, thus we only prove property (3) here.
	First of all we have the following expansion for the powers of $ \sigma $ near the horizon:
	\begin{align*}
	\sigma&=\frac{1}{\sqrt{\lambda}}-a s^\gamma +\frac{\beta^2 \sqrt{\lambda}}{2} s^{2\gamma}+O(s^{3\gamma})\\
	\sigma^2&=\dfrac{1}{\lambda}-\dfrac{2a}{\sqrt{\lambda}} s^\gamma+(a^2+\beta^2)s^{2\gamma}\\
	\sigma^3&=\frac{1}{\lambda\sqrt{\lambda}}-\frac{3a}{\lambda}s^\gamma+\frac{3(a^2+\frac{1}{2}\beta^2)}{\sqrt{\lambda}}s^{2\gamma}+O(s^{3\gamma})
	\end{align*}

	If we  plug   $ v_{a,\gamma} $ into Eq.\eqref{local}, then  each term  of $ \sigma^3\J[v_{a,\gamma}] $ will be:
	\begin{align*}
	\sigma^2 s^\gamma v_{a,\gamma}'(s)\theta^+
	=&(\dfrac{1}{\lambda}-\dfrac{2a}{\sqrt{\lambda}} s^\gamma+(a^2+\beta^2)s^{2\gamma})
	(-\dfrac{1}{\sqrt{\lambda}}+as^{\gamma})\theta^+\\
	=& (-\frac{1}{\lambda\sqrt{\lambda}}+\frac{3a}{\lambda}s^\gamma+O(s^{2\gamma}))\theta^+\\
	=&-\frac{1}{\sqrt{\lambda}}\beta s+3 a \beta s^{\gamma+1}-\frac{1}{\lambda\sqrt{\lambda}}(\theta^+-\lambda\beta s)\\ & +O(s^{3\gamma}+s^{2\gamma+1}+s^{\gamma+2}+s^3)\\
	\sigma^3+\sigma^2 s^\gamma v_{a,\gamma}'(s)
	=&\frac{1}{\lambda\sqrt{\lambda}}-\frac{3a}{\lambda}s^\gamma+\frac{3(a^2+\frac{1}{2}\beta^2)}{\sqrt{\lambda}}s^{2\gamma}\\
	& +(\dfrac{1}{\lambda}-\dfrac{2a}{\sqrt{\lambda}} s^\gamma+(a^2+\beta^2)s^{2\gamma})
	(-\dfrac{1}{\sqrt{\lambda}}+as^{\gamma})\\
	& +O(s^{3\gamma}+s^{2\gamma+1}+s^{\gamma+2}+s^3)\\	
	=&\frac{\beta^2}{2\sqrt{\lambda}}s^{2\gamma}+O(s^{3\gamma}+s^{2\gamma+1}+s^{\gamma+2}+s^3)\\
	\sigma s^{2\gamma}\beta^2 =&\frac{\beta^2}{\sqrt{\lambda}}s^{2\gamma}+O(s^{3\gamma})\\
	s^{3\gamma}\beta v_{a,\gamma}''(s) =&\frac{\beta\gamma}{\sqrt{\lambda}}s^{2\gamma-1}
	\end{align*}

	Put all these together we prove Eq.\eqref{expansion}.	
\end{proof}

By Eq.\eqref{expansion}, we can construct a family of super barriers (see  Lemma \ref{superlemma2}) and sub barriers (see Lemma \ref{sublemma} in the next section) for Jang's equation \eqref{jang} on a fixed domain near the horizon.

\begin{proposition}
	\label{superlemma2}
	There exist  constants $ a, s_2 $  which only depend on the local geometry of the initial data set near horizon, such that for $\forall 1<\gamma<\frac{5}{4} $,

	\begin{enumerate}
		
		\item $\J[v_{a,\gamma}(s)]\leq 0 $ for $ s\in(0,s_2] $.
		\item  $\frac{1}{\sqrt{\lambda}} \int_{s}^{s_2} \frac{1}{x^{\gamma}} dx+as-a s_2 +\sup_{q\in\Sigma}f(q,s_2)$ is a super barrier for the solution $ f $ of Eq.\eqref{jang} on $ (0,s_2] $.
	\end{enumerate}
\end{proposition}	
\begin{proof}
	We first prove Property (1). For the case $ \gamma>1 $,  rearrange Eq.\eqref{expansion} in Lemma \ref{taylor} on $ U_{\ov{s}} $  in ascending order. We denote $  c_1=\sup_{U_{\ov{s}}}(-\beta(k(\nu,\nu)+\frac{P}{2})) $ and recall that on $U_{\ov{s}}$ we have $ \lambda\beta s -\Lambda s^2<\theta^+[\Sigma_s]<\lambda\beta s +\Lambda s^2 $, then:
	\begin{align*}
	\sigma^3\J[v_{a,\gamma}]=&-\frac{\beta}{\sqrt{\lambda}} s+\frac{\beta\gamma}{\sqrt{\lambda}}s^{2\gamma-1}-\frac{1}{\lambda\sqrt{\lambda}} (\theta^+-\lambda\beta s)
	+3 a \beta s^{\gamma+1}\\
	&-\frac{\beta^2}{\sqrt{\lambda}}(k(\nu,\nu)+\frac{P}{2})s^{2\gamma}+O(s^3)\\ 
	\leq&-\frac{\beta}{\sqrt{\lambda}} s+\frac{\beta\gamma}{\sqrt{\lambda}}s^{2\gamma-1}+\frac{\Lambda}{\lambda\sqrt{\lambda}} s^2
	+3 a \beta s^{\gamma+1}
	+c_1\frac{\beta}{\sqrt{\lambda}} s^{2\gamma}+O(s^3)\\
	\end{align*}
	Thus 
	\begin{align*}
	-\sqrt{\lambda}\beta^{-1}\sigma^3\J[v_{a,\gamma}]\geq&
	s-\gamma s^{2\gamma-1}-\frac{\Lambda}{\lambda\beta} s^2
	-3 a \sqrt{\lambda} s^{\gamma+1}
	-c_1s^{2\gamma}+O(s^3)\\
	=&
	s(1-\gamma s^{2\gamma-2})-\left(\frac{\Lambda}{\lambda\beta} s^2
	+3 a \sqrt{\lambda} s^{\gamma+1}
	+c_1s^{2\gamma}\right)+O(s^3)\\
	=&
	s(1+\sqrt{\gamma} s^{\gamma-1})(1-\sqrt{\gamma} s^{\gamma-1})\\
	&-(1-\sqrt{\gamma} s^{\gamma-1})
	(\frac{\Lambda}{\lambda\beta} s^2+(\frac{\Lambda\sqrt{\gamma}}{\lambda\beta}+3a\sqrt{\lambda})s^{\gamma+1}		
	)\\
	&-(\frac{\Lambda\gamma}{\lambda\beta}+3a\sqrt{\lambda\gamma}+c_1)s^{2\gamma}+O(s^3)\\
	=&
	(1-\sqrt{\gamma} s^{\gamma-1})
	(s+\sqrt{\gamma} s^{\gamma}
	-\frac{\Lambda}{\lambda\beta} s^2-(\frac{\Lambda\sqrt{\gamma}}{\lambda\beta}+3a\sqrt{\lambda})s^{\gamma+1}
	)\\
	&
	+(-3a\sqrt{\lambda\gamma}-\frac{\Lambda\gamma}{\lambda\beta}-c_1)s^{2\gamma}+O(s^3)
	\end{align*}
	Let's look at the last line of the above calculation. Notice the following facts:
	\begin{enumerate}
		\item   $ 1-\sqrt{\gamma} s^{\gamma-1}\geq 0 $ for $ \forall \gamma>1 $ and $ \forall s\in[0,e^{-\frac{1}{2}}] $.
		\item Set $ a=-\sup_{q\in\Sigma}\dfrac{\frac{5\Lambda}{4\lambda\beta}+|c_1|+1}{3\sqrt{\lambda}} $, then $ -3a\sqrt{\lambda\gamma}-\frac{\Lambda\gamma}{\lambda\beta}-c_1\geq1>0$ for $\gamma\in(1,\frac{5}{4}) $. Notice that $ a $ is not related with $ \gamma $.
		\item  Because we set $ a $ in terms of the local geometry data near the horizon, there exists constant $ s_2^*>0 $ which only depends on the same local geometry (not related with $ \gamma $), such that 
		\begin{align*}
		&  s+\sqrt{\gamma} s^{\gamma}
		-\frac{\Lambda}{\lambda\beta} s^2-(\frac{\Lambda\sqrt{\gamma}}{\lambda\beta}+3a\sqrt{\lambda})s^{\gamma+1}\\
		\geq&s+  s^{\frac{5}{4}}
		-\frac{\Lambda}{\lambda\beta} s^2-\big|\frac{\sqrt{5}\Lambda }{2\lambda\beta}+3a\sqrt{\lambda}\big|s^2
		>0 
		\end{align*}
		
		holds for $ s\in(0,s_2^*] $. This is obvious because the coefficients of the leading terms $ s $ and $ s^\frac{5}{4} $ are both positive constants, and their orders are at least $ \frac{3}{4} $ away from the orders of the rest terms, whose coefficients are bounded and only depend  on the local geometry near the horizon of the initial data set.	
	\end{enumerate}
	
	Then  on $ (0,\min(s_2^*, e^{-\frac{1}{2}}) ] $ we have:
	\begin{align*}
	-\sqrt{\lambda}\beta^{-1}\sigma^3\J[v_{a,\gamma}] \geq&(-3a\sqrt{\lambda\gamma}-\frac{\Lambda\gamma}{\lambda\beta}-c_1)s^{2\gamma}+O(s^3)\\
	\geq&s^{2\gamma}+O(s^3)\\
	\geq&s^{\frac{5}{2}}+O(s^3) 
	\end{align*}

	Therefore, there exists constant $ s_2 $ only depending on the local geometry of the initial data, such that  $ -\sqrt{\lambda}\beta^{-1}\sigma^3\J[v_{a,\gamma}]\geq 0 $ for $ s\in(0,s_2], $
	and thus $ \J[v_{a,\gamma}]\leq 0 $ because $ \lambda>0 $ and $ \beta $ is positive function. Thus 
	$v_{a,\gamma}(s) $ is a supersolution of Jang's equation \eqref{jang} on $ (0,s_2] $.\par
	
	For Property (2) we only need to notice that for $ \gamma>1 $, $ v_{a,\gamma} $ blows up at the  order of $ s^{1-\gamma} $, which is much faster than $ -c\log s $ for any constant $ c $ when $ s\to 0 $. In the previous section we prove that any solution $ f $ cannot blowup faster than $ -c\log s $. Thus for $ \gamma>1 $, $ v_{a,\gamma} $ always blows up faster than any solution $ f $ of Eq.\eqref{jang} near the horizon. Therefore  a vertical translation of $ v_{a,\gamma} $ at $ s=s_2 $ is enough to make it   a super barrier.	
\end{proof}

By the comparison principle, c.f. \cite[Chapter 10]{Gilbarg}, we have the following   result:

\begin{proposition}
	There exist  constants $ a, s_2 $  which only depend on the local geometry of the initial data set near horizon, such that for $\forall 1<\gamma<\frac{3}{2} $, and 
	any function  $ f $ satisfies Jang's equation \eqref{jang} on $ U_{\ov{s}} $ and $ f\to +\infty $ when $ s\to 0 $, the following inequality :
	\begin{equation}
	\label{supergamma}
	f(p,s) \leq \sup_{q\in\Sigma}f(q,s_2)+\frac{1}{\sqrt{\lambda}} \int_{s}^{s_2} \frac{1}{x^{\gamma}} dx +as-a s_2
	\end{equation}
	holds for any $ p\in\Sigma$ and $ s\in(0,s_2] $. 
\end{proposition}

We let $ \gamma\to 1 $ in Ineq.\eqref{supergamma}. Then because $ a, s_2 $ is not related to $ \gamma $, we have the  following estimate for the upper bound of blowup rate:
\begin{proposition}
	There exist  constants $ a, s_2 $  which only depend on the local geometry of the initial data set near horizon, such that 
	any function  $ f $ satisfies Jang's equation \ref{jang} on $ U_{\ov{s}} $ and $ f\to +\infty $ when $ s\to 0 $, the following inequality :
	\begin{equation}
	\label{super2}
	f(p,s) \le \sup_{q\in\Sigma}f(q,s_2)-\frac{1}{\sqrt{\lambda}} \log s+\frac{1}{\sqrt{\lambda}} \log s_2  +as-a s_2
	\end{equation}
	holds for any $ p\in\Sigma$ and $ s\in(0,s_2] $.

\end{proposition}

Thus the first half of Theorem \ref{thm1} is proved.

\subsection{Sub Estimate}

By Eq.\eqref{expansion}, we can construct a family of sub barriers in a similar way as in the previous section for Jang's equation \eqref{jang} on a fixed domain near the horizon.

\begin{proposition}
	\label{sublemma}
	There exist  constants $ a, s_3 $  which only depend on the local geometry of the initial data set near horizon, such that for $\forall \frac{3}{4}<\gamma<1 $,

	\begin{enumerate}
		
		\item $\J[v_{a,\gamma}(s)]\geq 0 $ for $s\in (0,s_3] $.
		\item  $\frac{1}{\sqrt{\lambda}} \int_{s}^{s_3} \frac{1}{x^{\gamma}} dx+as-a s_3 +\inf_{q\in\Sigma}f(q,s_3)$ is a sub barrier for solution $ f $ of Eq.\eqref{jang} on $ (0,s_3] $.
	\end{enumerate}
\end{proposition}	
\begin{proof}
	We first prove Property (1). For the case $ \gamma<1 $,  rearrange the terms of Eq.\eqref{expansion}    in ascending order on $ U_{\ov{s}} $. Denote $  c_2=\inf_{U_{\ov{s}}}(-\beta(k(\nu,\nu)+\frac{P}{2})) $ and recall that on $U_{\ov{s}}$ we have $ \lambda\beta s -\Lambda s^2<\theta^+[\Sigma_s]<\lambda\beta s +\Lambda s^2 $, then:
	\begin{align*}
	\sigma^3\J[v_{a,\gamma}] =&\frac{\beta\gamma}{\sqrt{\lambda}}s^{2\gamma-1}-\frac{\beta}{\sqrt{\lambda}} s-\frac{\beta^2}{\sqrt{\lambda}}(k(\nu,\nu)+\frac{P}{2})s^{2\gamma}
	+3 a \beta s^{\gamma+1}\\
	&-\frac{1}{\lambda\sqrt{\lambda}} (\theta^+-\lambda\beta s)
	+O(s^{3\gamma})\\ 
	\geq&\frac{\beta\gamma}{\sqrt{\lambda}}s^{2\gamma-1}-\frac{\beta}{\sqrt{\lambda}} s
	+c_2\frac{\beta}{\sqrt{\lambda}} s^{2\gamma}+3 a \beta s^{\gamma+1}
	-\frac{\Lambda}{\lambda\sqrt{\lambda}} s^2+O(s^{3\gamma})\\
	\end{align*}
	Thus 
	\begin{align*}
	\sqrt{\lambda}\beta^{-1}\sigma^3\J[v_{a,\gamma}] \geq&
	\gamma s^{2\gamma-1}-s+c_2s^{2\gamma}
	+3 a \sqrt{\lambda} s^{\gamma+1}-\frac{\Lambda}{\lambda\beta} s^2
	+O(s^{3\gamma})\\
	=&
	\gamma s^{2\gamma-1}(1 -\frac{1}{\gamma}s^{2-2\gamma})+\left(c_2s^{2\gamma}
	+3 a \sqrt{\lambda} s^{\gamma+1}-\frac{\Lambda}{\lambda\beta} s^2\right)+O(s^{3\gamma})\\
	=&
	\gamma s^{2\gamma-1}(1 -\frac{1}{\sqrt{\gamma}}s^{1-\gamma})(1 +\frac{1}{\sqrt{\gamma}}s^{1-\gamma})\\
	&+(1 -\frac{1}{\sqrt{\gamma}}s^{1-\gamma})
	(c_2 s^{2\gamma}+(3a\sqrt{\lambda}+\frac{c_2}{\sqrt{\gamma}})s^{\gamma+1}		
	)\\
	&+(3a\sqrt{\frac{\lambda}{\gamma}}-\frac{\Lambda}{\lambda\beta}+\frac{c_2}{ \gamma })s^{2}+O(s^{3\gamma})\\
	=&
	(1 -\frac{1}{\sqrt{\gamma}}s^{1-\gamma})
	(\gamma s^{2\gamma-1}+\sqrt{\gamma} s^{\gamma}+c_2 s^{2\gamma}+(3a\sqrt{\lambda}+\frac{c_2}{\sqrt{\gamma}})s^{\gamma+1}	
	)\\
	&
	+(3a\sqrt{\frac{\lambda}{\gamma}}-\frac{\Lambda}{\lambda\beta}+\frac{c_2}{ \gamma })s^{2}+O(s^{3\gamma})
	\end{align*}
	Let's look at the last line of the above calculation. Notice the following facts:
	\begin{enumerate}
		\item   $ 1 -\frac{1}{\sqrt{\gamma}}s^{1-\gamma}\geq 0 $ for $ \forall \frac{3}{4}<\gamma<1 $ and $ \forall s\in[0,\frac{9}{16}] $.
		\item Set $ a= \sup_{q\in\Sigma}\frac{2}{3\sqrt{3\lambda}}(1+\frac{4}{3}|c_2|+\frac{\Lambda}{\lambda\beta})$, then $ 3a\sqrt{\frac{\lambda}{\gamma}}-\frac{\Lambda}{\lambda\beta}+\frac{c_2}{ \gamma }\geq1>0$ for $\gamma\in( \frac{3}{4},1) $. Notice that $ a $ is not related with $ \gamma $.
		\item  Because we set $ a $ in terms of the local geometry data near the horizon,  there exists constant $ s_3^*>0 $ which only depends on the same local geometry (not related with $ \gamma $), such that 
		\begin{align*}
		&  \gamma s^{2\gamma-1}+\sqrt{\gamma} s^{\gamma}+c_2 s^{2\gamma}+(3a\sqrt{\lambda}+\frac{c_2}{\sqrt{\gamma}})s^{\gamma+1}\\
		\geq&
		\frac{3}{4}
		s +\frac{\sqrt{3}}{2} s -|c_2| s^{\frac{3}{2}}-\big|3a\sqrt{\lambda}+\frac{c_2}{\sqrt{\gamma}}\big|s^{\frac{7}{4}}
		>0 
		\end{align*}
		
		holds for $ s\in(0,s_3^*] $. This is obvious because the coefficients of the leading terms $ s $   are both positive constants, and their orders are at least $ \frac{1}{2} $ away from the orders of the rest terms, whose coefficients are bounded and only depend  on the local geometry near the horizon of the initial data set.	
	\end{enumerate}
	
	Then  on $ (0,\min(s_3^*,  \frac{9}{16})  ] $ we have:
	\begin{align*}
	\sqrt{\lambda}\beta^{-1}\sigma^3\J[v_{a,\gamma}]\geq&(3a\sqrt{\frac{\lambda}{\gamma}}-\frac{\Lambda}{\lambda\beta}+\frac{c_2}{ \gamma })s^{2}+O(s^{3\gamma})\\
	\geq&s^{2 }+O(s^{3\gamma})\\
	\geq&s^2+O(s^\frac{9}{4}) 
	\end{align*}

	Therefore, there exists constant $ s_3 $ only depending on the local geometry of the initial data, such that  $  \sqrt{\lambda}\beta^{-1}\sigma^3\J[v_{a,\gamma}]\geq 0 $ for $ s\in(0,s_3], $
	and thus $ \J[v_{a,\gamma}]\geq 0 $ because $ \lambda>0 $ and $ \beta $ is positive function. Thus 
	$v_{a,\gamma}(s) $ is a supersolution of Jang's equation \eqref{jang} on $ (0,s_3] $.\par
	
	For Property (2) we only need to notice that for $ \frac{3}{4}<\gamma<1 $, $ v_{a,\gamma} $ is finite at $ s=0 $. By the fact that $ f $  blows up   near the horizon,  a vertical translation of $ v_{a,\gamma} $ at $ s=s_3 $ is enough  to make it   a sub barrier.
\end{proof}

By the comparison principle, c.f. \cite[Chapter 10]{Gilbarg}, we have the following  result:

\begin{proposition}
	There exist  constants $ a, s_3 $  which only depends on the local geometry of the initial data set near horizon, such that for $\forall \frac{3}{4}<\gamma<1$, and 
	any function  $ f $ satisfies Jang's equation \eqref{jang} on $ U_{\ov{s}} $ and $ f\to +\infty $ when $ s\to 0 $, the following inequality:
	\begin{equation}
	\label{subgamma}
	f(p,s) \geq \inf_{q\in\Sigma}f(q,s_3)+\frac{1}{\sqrt{\lambda}} \int_{s}^{s_3} \frac{1}{x^{\gamma}} dx +as-a s_3
	\end{equation}
	holds for any $ p\in\Sigma$ and $ s\in(0,s_3] $. 
\end{proposition}

Let $ \gamma\to 1 $ in Ineq.\eqref{subgamma}, then because $a,  s_3 $ is not related to $ \gamma $, we have the  following estimate for the lower bound of blowup rate:
\begin{proposition}
	There exist  constants $ a, s_3 $  which only depend on the local geometry of the initial data set near horizon, such that for 
	any function  $ f $ satisfies Jang's equation \ref{jang} on $ U_{\ov{s}} $ and $ f\to +\infty $ when $ s\to 0 $, the following inequality :
	\begin{equation}
	\label{sub}
	f(p,s) \geq \inf_{q\in\Sigma}f(q,s_3)-\frac{1}{\sqrt{\lambda}} \log s+\frac{1}{\sqrt{\lambda}} \log s_3  +as-a s_3
	\end{equation}
	holds for any $ p\in\Sigma$ and $ s\in(0,s_3] $.

\end{proposition}

Thus  the second half of   Theorem \ref{thm1} is proved.

Put together our barriers, we prove that in  Foliation B, there exist  constants $ a, s_0 $  which only depend on the local geometry of the initial data set near horizon, such that for 
any function  $ f $ satisfies Jang's equation \eqref{jang} on $ U_{\ov{s}} $ and $ f\to +\infty $ when $ s\to 0 $, the following inequality :
\begin{equation}\label{all}
\inf_{q\in\Sigma}f(q,s_0)-\frac{1}{\sqrt{\lambda}} \log \frac{s}{s_0}  +a(s-s_0)\leq f(p,s) \leq \sup_{q\in\Sigma}f(q,s_0)-\frac{1}{\sqrt{\lambda}} \log \frac{s}{s_0}   -a(s-s_0)
\end{equation}
holds for any $ p\in\Sigma$ and $ s\in(0,s_0] $. 
Because $ s $ is comparable with the distance to the horizon, thus the above inequality implies the following $ C^0 $ lower and upper bound:
\begin{proposition} 
	There exist  constants $ a, \tau_0 $  which only depend on the local geometry of the initial data set near horizon, such that for
	any function  $ f $ satisfies Jang's equation \eqref{jang} on $ V_{\tau_0} $ and $ f\to +\infty $ when $ \tau\to 0 $, $ |f+\frac{1}{\sqrt{\lambda}} \log \tau| $ is a bounded function in $ V_{\tau_0} $. Here $ \tau  $ is the distance function to horizon, and $ V_{\tau_0}=\{x\in\M|dist(x,\Sigma)\in(0,\tau_0]\} $

\end{proposition}

It is a direct consequence of our barrier arguments that the following inequality holds for $ \forall 0<s\leq s_1\leq s_0 $:
\begin{equation}\label{main}
\inf_{q\in\Sigma}f(q,s_1)-\frac{1}{\sqrt{\lambda}} \log \frac{s}{s_1}  +a(s-s_1)\leq f(\cdot,s) \leq \sup_{q\in\Sigma}f(q,s_1)-\frac{1}{\sqrt{\lambda}} \log \frac{s}{s_1}   -a(s-s_1)
\end{equation}
As a consequence of Ineq.\eqref{main}, suppose $ x_1, x_2 $ are the points on $ \Sigma_{s_1} $ such that $ f(x_1,s_1)=\sup_{q\in\Sigma}f(q,s_1) $ and $ f(x_2,s_1)=\inf_{q\in\Sigma}f(q,s_1) $, then:
\begin{align*}
\partial_sf(x_1,s_1)
=&\lim_{s\to s_1}\frac{f(x_1,s)-f(x_1,s_1)}{s-s_1}\\
=&\lim_{s\to s_1^-}\frac{f(x_1,s_1)-f(x_1,s)}{s_1-s }\\
\geq&\lim_{s\to s_1^-}\frac{f(x_1,s_1)-\sup_{q\in\Sigma}f(q,s_1)+\frac{1}{\sqrt{\lambda}} \log \frac{s}{s_1}   +a(s-s_1)}{s_1-s }\\
=&\lim_{s\to s_1^-}\frac{ -\frac{1}{\sqrt{\lambda}} (\log s_1-\log s)   -a(s_1-s )}{s_1-s }\\
=&-\frac{1}{\sqrt{\lambda}s_1}-a
\end{align*}

Similarly we can also get:
\begin{equation*}
\partial_sf(x_2,s_1)\leq-\frac{1}{\sqrt{\lambda}s_1}+a
\end{equation*}
Thus Theorem \ref{thm1} gets proved.

\section{Gradient Estimates}\label{gradient}

In this section we prove some   gradient estimates for  the blowup solution $ f $ of Jang's equation Eq.\eqref{jang}. First we need the following a priori estimate for Jang's equation by R. Schoen and S. -T. Yau:

\begin{theorem}\cite[Proposition  2]{yau}
	Let $ F(x) $ be a given $ C^2 $ function on $ \M $ and suppose $ \mu_1,\mu_2,\mu_3 $ are constants so that 
	$$\sup_{\M} |F|\leq \mu_1,\mbox{ }\mbox{ }\mbox{ } \sup_{\M} |\nabla F|\leq \mu_2,\mbox{ }\mbox{ }\mbox{ } \sup_{\M} |\nabla^2 F|\leq \mu_3$$
	and $ f $ is a $ C^3 $ solution of 
	$$(g^{ij}-\frac{f^i f^j}{1+|\nabla f |^2})(\frac{\nabla_{ij}f}{\sqrt{1+|\nabla f |^2}}-k_{ij})=F
	$$
	Denote $ N=\f $. If $ X_0\in N $, let $ (y^1,y^2,y^3,y^4) $ be normal coordinates in $ \M\times\R $ centered at $ X_0 $ so that the tangent space to $ N $ at $ X_0 $ is the $ y^1y^2y^3 $-space. Then in a neighborhood of $ X_0 $, $ N $ is given by the graph of a function $ w(y) $, $ y=(y^1, y^2, y^3) $. We called this the local defining function $ w $ for $ N $.\par
	
	Then, there is a constant $ \rho>0 $ depending only on the initial data and $ \mu_1,\mu_2,\mu_3 $ so that for any $ X_0\in N $, the local defining function $ w $ for $ N $ is defined on $ \{|y|\leq\rho\} $, and satisfies for any $ \alpha\in(0,1) $,
	$$\sup_{|y|\leq\rho}(|w(y)|+|\partial w(y)|+|\partial\partial w(y)|+|\partial\partial\partial w(y)|+|\partial\partial\partial w(y)|_{\alpha,\rho})\leq c_1(\alpha),$$
	
	where $ c_1 $ depends only on $ \alpha $, the initial data, and $ \mu_1,\mu_2,\mu_3 $. Moreover, we may require 
	$$ N\cap B^4_\rho(X_0)\subset \{Y:y^4=w(y)\} $$
	We also have the following Harnack-type inequalities:
	\begin{align*}
	&\sup_{N\cap B^4_\rho(X_0)}<\ov{e}_4,\nu> \leq   c_2 \inf_{N\cap B^4_\rho(X_0)}<\ov{e}_4,\nu>\\
	&\sup_{N\cap B^4_\rho(X_0)}|\ov{\nabla}\log<\ov{e}_4,\nu>| \leq  c_3
	\end{align*}
	where $ \ov{e}_4 $ is the downward unit normal to $ N $, $ \nu $ is the downward unit parallel vector field tangent to the $ \R $ factor, $ \ov{\nabla} $ is the covariant derivative on $ N=\f $, and $ c_2, c_3 $ are constants only depending on the initial data and $ \mu_1,\mu_2,\mu_3 $.

\end{theorem}

Then following the arguments in \cite[Proposition 4]{yau}, we know that $ N=\f $ converges to the cylinder $ \Sigma\times\R $. In fact, from the fact that $  f-a $ is also a solution of Eq.\eqref{jang}, by the estimates of \cite[Proposition 2]{yau} there is a sequence $ \{a_i\} $ tending to infinity such that   the graph of $ f-a_i $ converges smoothly on compact subsets of $ \M\times\R $ to a limiting 3-dim submanifold of $ \M\times\R $. Then Harnack inequality in \cite[Proposition 2]{yau} implies that this limiting manifold is $ \Sigma\times\R $.  Let $ U $ be a neighborhood of $ \Sigma $. Following the arguments in \cite[Proposition 4, Corollary 2]{yau}, we can define a coordinate system on the neighborhood $U\times\R$ of $ \Sigma\times\R $ in $ \M\times\R $ by taking the fourth coordinate $ \tau $ to be the distance function to $ \Sigma\times\R $ in $ \M\times\R $. Let $\bar \Psi' : \Sigma \times(-\epsilon',\epsilon')\to \M$ be the map
\begin{equation*}
\bar\Psi'
:
\Sigma\times (-\epsilon',\epsilon')\times \R \to \M \times\R
:
(p,\tau,z) \mapsto \big( \exp_p(\tau\nu), z \big).    
\end{equation*}
Thus $ \bar{\Psi}' $ is compatible with Foliation A.

Then for a function $h$ on $C_{\bar z}$ we let $graph_{\bar\Psi'} h$ be the set
\begin{equation*}
graph_{\bar\Psi'} h = \{ \bar\Psi(p, h(p,z),z) : (p,z) \in \Sigma\times \R \}.
\end{equation*}

Therefore, it is a direct consequence of the above reasoning that there exists constant  $ \bar{z} $, such that $N \cap(U\times[\bar{z},\infty))  $ can be written as a graph of $ u $ on $ C_{\bar{z}}:=\Sigma\times[\bar{z},\infty) $. Furthermore, by translating the $ C^0 $ barriers of $ f  $ Eq.\eqref{thm barrier} into $ C^0 $ barriers of $h  $, we get the following:

\begin{proposition}
	
	Under the same assumptions as  Theorem \ref{thm1}, 
	for each solution $ f $ of Eq.\eqref{jang} which blows up at $ \Sigma $, there exist positive constants $\bar z $  and $C_0' $, $C_2' $,
	such that $graph f \cap U
	\times [\bar z,\infty)$ can be written as the graph of a function
	$h$ over $C_{\bar z}:= \Sigma  \times [\bar z,\infty)$ under coordinate $ \bar{\Psi}' $,  and
	\begin{equation*}
	C_2'\exp({-\sqrt{\lambda } z} ) \leq|h(p,z)|   \leq C_0'\exp({-\sqrt{\lambda } z})
	\end{equation*}
	
\end{proposition}

\subsection{$C^2$ Bound}

In this section, we prove a refinement of \cite[Corollary 2]{yau} and \cite[Theorem 4.2]{metzger}.  \cite[Corollary 2]{yau} states that $ N=\f $ converges uniformly in $ C^2 $ to the cylinder $ \Sigma\times\R $ for large values of $ f $, and \cite[Theorem 4.2]{metzger} states that the rate of this convergence is exponential. Recall the statements of \cite[Corollary 2]{yau} and \cite[Theorem 4.2]{metzger}:

\begin{theorem}\cite[Corollary 2]{yau}
	\label{thm:sideways1}
	Assume the conditions of Theorem \ref{original} and let $\Sigma$ be a connected component
	of the apparent horizons, on which $ f $ tends to $ +\infty $ ($ -\infty $ respectively). Let $U$ be a neighborhood of $\Sigma$ with positive
	distance to any other apparent horizons in $ \M $.
	
	Then for all $\epsilon>0$ there exists $\bar z = \bar z(\epsilon)$,
	depending also on the geometry of $(\M,g,k)$, such that $N \cap( U
	\times [\bar z,\infty))$ can be written as the graph of a function
	$u$ over $C_{\bar z}:= \Sigma \times [\bar z,\infty)$, so that
	\begin{equation*}
	|u(p,z)| + |^{C_{\bar z}}\nabla u(p,z) | + |^{C_{\bar z}}\nabla^2 u(p,z)| < \epsilon.
	\end{equation*}
	for all $(p,z)\in C_{\bar z}$. Here, $^{C_{\bar z}}\nabla$ denotes
	covariant differentiation along $C_{\bar z}$.
\end{theorem}

\begin{theorem}\cite[Theorem 4.2]{metzger}
	\label{thm:exponential1}
	Let $N_0 = \f_0$ be the manifold constructed in Theorem \ref{thm:blowup} and 
	assume the situation of Theorem~\ref{thm:sideways}. Then there exists $\bar z = \bar z(\epsilon)$,
	depending also on the geometry of $(\M,g,k)$, such that $N_0 \cap (U
	\times [\bar z,\infty))$ can be written as the graph of a function
	$u_0$ over $C_{\bar z}:= \Sigma \times [\bar z,\infty).$
	If in addition $\Sigma$ is
	strictly stable with principal eigenvalue $\lambda >0$, then
	for all $\delta < \sqrt{\lambda}$ there exists $c=c(\delta)$ depending only
	on the data $(\M,g,k)$ and $\delta$ such that
	\begin{equation*}
	|u_0(p,z)| + |^{C_{\bar z}}\nabla u_0(p,z)| + |^{C_{\bar z}}\nabla^2 u_0(p,z)| \leq c\exp(-\delta z).
	\end{equation*}
	where  $^{C_{\bar z}}\nabla$ is the covariant derivative w.r.t. the induced metric on $C_{\bar{z}}$. 
\end{theorem}

In this section, we are able to prove that the rate of this convergence is $ \exp(-\sqrt{\lambda} z) $.
In \cite[Section 4]{metzger}, J. Metzger provided a way to extend the $ C^0 $ super  control to $ C^2 $ super  control. The procedure presented here is the same as   in \cite[Theorem 4.2]{metzger}. We briefly outline here for consistence. \par

In \cite[Theorem 4.2]{metzger}, it is computed that  the value of Jang's operator for a function $h$ on $ \Sigma\times\R $ is the
following:
\begin{align*}
\label{eq:8}
\J [h]
=&
\partial_z^2 h
+ \gamma_{h(p,z)}^{ij} \nabla^2_{i,j}h
- 2 \gamma_{h(p,z)}^{ij} \partial_i h\cdot k(\partial_s, \partial_j)
- \theta^+[\Sigma_{h(p,z)}]\\
& + Q(h,^{C_{\bar z}}\nabla h,^{C_{\bar z}}\nabla^2 h)
\end{align*}
where $\gamma_s$ is the metric on $\Sigma_s$ and $Q$ is of the form
\begin{equation*}
Q(h,^{C_{\bar z}}\nabla h,^{C_{\bar z}}\nabla ^2 h)
=
h * ^{C_{\bar z}}\nabla 
+ ^{C_{\bar z}}\nabla  * ^{C_{\bar z}}\nabla 
+ ^{C_{\bar z}}\nabla  * ^{C_{\bar z}}\nabla  *^{C_{\bar z}} \nabla ^2 h
\end{equation*}
where $*$ denotes some contraction with a bounded
tensor. The vectors $\partial_i$, $i=1,2$ denote directions
tangential to $\Sigma$ and $\partial_z$ the direction along the
$\R$-factor in $C_{\bar z}$.\par

By \cite[Corollary 2]{yau}, we know that $ Q(h,^{C_{\bar z}}\nabla h,^{C_{\bar z}}\nabla ^2 h) $ is low order term. Thus by freezing coefficients, $h$ satisfies a linear, uniformly elliptic equation of the form
\begin{equation*}
a^{ij} \partial_i\partial_j h + < b, ^{C_{\bar z}}\nabla  h> - \theta^+[\Sigma_{h(p,z)}] = F.
\end{equation*}
Notice that because  $\theta^+[\Sigma_\tau] =\tau L_\Sigma 1  +O(\tau^2)$,  $\theta^+[\Sigma_{h(p,z)}]$ also decays
exponentially with the order $ e^{-\sqrt{\lambda}z} $.\par

Now using standard Schauder interior estimates for linear elliptic equations (c.f. \cite[Chapter 6]{Gilbarg}), we can bound $ C^2 $ norm of $ h$ by its $ C^0 $ norm and  the norm of $\theta^+[\Sigma_{h(p,z)}]$. By the fact that they both decay exponentially with the order $ e^{-\sqrt{\lambda}z} $, we can conclude that $ C^2 $ norm of $ h $ also decays exponentially with the order $ e^{-\sqrt{\lambda}z} $. Thus we prove the first part of Theorem \ref{thm2}.

\begin{theorem}
	\label{ex1}
	Under the same assumptions of  Theorem \ref{thm1}, 
	for each solution $ f $ of Eq.\eqref{jang} which blows up at $ \Sigma $, there exist positive constants $\bar z $  and $C_1' $, $C_2' $,
	such that $graph f \cap U
	\times [\bar z,\infty)$ can be written as the graph of a function
	$h$ over $C_{\bar z}:= \Sigma  \times [\bar z,\infty)$ under coordinate system $ \bar{\Psi}' $,  and
	\begin{equation*}
	|h(p,z)| + |^{C_{\bar z}}\nabla  h(p,z)| + |^{C_{\bar z}}\nabla ^2 h(p,z)| \leq C_1'\exp({-\sqrt{\lambda } z})
	\end{equation*}
	and
	\begin{equation*}
	|h(p,z)|\geq	C_2'\exp({-\sqrt{\lambda } z} ) 
	\end{equation*}
	where $^{C_{\bar z}}\nabla $ is the covariant derivative w.r.t. the induced metric on $C_{\bar{z}}$. 
\end{theorem}

For the convenience of the next section, we prefer to rewrite  the above theorems under a coordinate system in $ U\times\R $ which is compatible with Foliation B.

Let $\bar \Psi : \Sigma \times(-\epsilon,\epsilon)\to \M$ be the map
\begin{equation*}
\bar\Psi
:
\Sigma\times (-\epsilon,\epsilon)\times \R \to \M \times\R
:
(p,s,z) \mapsto \big( \Psi_B(p,s), z \big).    
\end{equation*}
Thus $ \bar{\Psi} $ is compatible with Foliation B.

Then for a function $u$ on $C_{\bar z}$ we let $graph_{\bar\Psi} u$ be the set
\begin{equation*}
graph_{\bar\Psi} u = \{ \bar\Psi(p, u(p,z),z) : (p,z) \in \Sigma\times \R \}.
\end{equation*}

For a small neighborhood of $ \Sigma $,  there exist constants $ \alpha_i, i=1,2,3 $, such that the following inequalities:
\begin{gather*}
\frac{  s}{ \tau}\in(\alpha_1^{-1},\alpha_1)\\
\frac{\partial s}{\partial\tau}, \frac{\partial \tau}{\partial s}\in(\alpha_2^{-1},\alpha_2)\\
\frac{\partial^2 s}{\partial\tau^2},\frac{\partial^2 \tau}{\partial s^2}\in(\alpha_3^{-1},\alpha_3)
\end{gather*}
holds in this neighborhood.

Then we can rewrite the above theorem under coordinates defined by $ \Psi $:

\begin{theorem}
	\label{ex}
	Under the same assumption as  Theorem \ref{thm1}, 
	for each solution $ f $ of Eq.\eqref{jang} which blows up at $ \Sigma $, there exist positive constants $\bar z $  and $C_1 $, $C_2 $,
	such that $graph f \cap U
	\times [\bar z,\infty)$ can be written as the graph of a function
	$u$ over $C_{\bar z}:= \Sigma  \times [\bar z,\infty)$ under coordinate system $ \bar{\Psi} $,  and
	\begin{equation*}
	|u(p,z)| + |^{C_{\bar z}}\nabla  u(p,z)| + |^{C_{\bar z}}\nabla ^2 u(p,z)| \leq C_1\exp({-\sqrt{\lambda } z})
	\end{equation*}
	and
	\begin{equation*}
	|u(p,z)|\geq	C_2\exp({-\sqrt{\lambda } z} ) 
	\end{equation*}
	where $^{C_{\bar z}}\nabla $ is the covariant derivative w.r.t. the induced metric on $C_{\bar{z}}$. 
\end{theorem}

In the above theorem, the constants $ C_1,C_2,\bar z $ are of course related to the solution $ f $ because Jang's equation is invariant under vertical translation. By the fact that super barrier and sub barrier  depend on $ \sup_{\Sigma_{s_0}} f $ and $  \inf_{\Sigma_{s_0}} f $,     we know  constants $ C_1,C_2,\bar z $ also depend on these quantities. \par

\subsection{Gradient Estimate}

Now under the coordinates defined by $ \bar \Psi $, the hypersurface $ N $ can be expressed in two different way in a neighborhood of $ \Sigma $: either as the graph of $ f $ on $ \M $, or as the graph of $ u $ on cylinder $ \Sigma\times\R $:
\begin{align*}
z=f(p,s)\\
s=u(p,z)
\end{align*}
hold  for $ p\in\Sigma $.\par

Denote $ F(p,z,s)=f(p,s)-z $. Then $ F(p,z,u(p,z))=0 $. Let $ e_i, i=1,2 $  be the orthonormal frame on $ \Sigma $. Parallel translate them to obtain an orthonormal frame on $ \Sigma_s $, and still denote them as $ e_i, i=1,2 $. Denote $ \partial_i, i=1,2 $ to be the derivatives taken in the direction  of $ e_i, i=1,2 $, respectively,  and $ \partial_s$ to be the derivative taken w.r.t the parameter $ s $ of Foliation B.  Then by Implicit Function Theorem, 
\begin{align}
\partial_z u(p,z)=&-\big([J_{F,s}(p,z,s)]^{-1}\partial_z F(p,z,s)\big)\big|_{s=u(p,z)}\\
\partial_i u(p,z)=&-\big([J_{F,s}(p,z,s)]^{-1}\partial_i F(p,z,s)\big)\big|_{s=u(p,z)}\label{tangential}
\end{align}
Thus we have 
\begin{equation}
\partial_s f(p,s)\big|_{s=u(p,z)}=J_{F,s}(p,z,s )\big|_{s=u(p,z)}= \frac{1}{\partial_z u(p,z)}
\end{equation}

Therefore, by the fact that $$ |\partial_z u(p,z)|\leq|^{C_{\bar z}}\nabla  u (p,z)|\leq C_1 e^{-\sqrt{\lambda}z}, $$ and 
$$ s=  u(p,z)  \geq C_2 e^{-\sqrt{\lambda}z}, $$
we have the following lower bound for $ |\partial_s f  | $:
\begin{equation*}
| \partial_s f(p,s)|
\geq\frac{1}{C_1 e^{-\sqrt{\lambda}z}}
\geq \frac{C_2}{C_1  s} 	 
\end{equation*}

We now prove the upper bound for the gradient estimate for $ f $. As discussed above, it is equivalent to find the lower bound for $ |\partial_z u| $. We first prove the following  weaker version:

\begin{proposition}
	There exist   constants $ C>0 $ and $z_1\geq\bar{z}$ such that for $ \forall z> z_1 $, there exists a point $ p_z\in\Sigma $, such that 
	$$ \partial_z u(p_z,z):=\lim_{\delta z\to0}\frac{u(p_z,z+\delta z)-u(p_z,z)}{\delta z}\leq -C \exp(-\sqrt{\lambda}z) $$
	
\end{proposition}
\begin{proof}
	For an arbitrary $ z_0\geq\bar{z} $, denote $ p_0\in\Sigma $ to be the point where $ u $ achieves its maximum when $ z=z_0 $, and   $ s_0=\max_{p\in\Sigma}(u(p,z_0))=u(p_0,z_0) $. Then on $ \Sigma_{s_0} $  under Foliation B, $ f $ achieves its maximum on $ \Sigma_{s_0} $ at $ p_0\in\Sigma $. By the gradient estimate of $ f $ at maximum point on $ \Sigma_{s_0} $ in Theorem \ref{thm1}, we have 
	$$|\partial_s f(p_0,s_0)|\leq \frac{1}{\sqrt{\lambda}s_0}+a$$
	Translate it back to the gradient estimate for $ u $ at $ (p_0,z_0) $ by Implicit Function Theorem, we have:
	\begin{align*}
	|\partial_z u(p_0,z_0)|=&\frac{1}{|\partial_s f(p_0,s_0)|}\geq\frac{\sqrt{\lambda}s_0}{1+a \sqrt{\lambda}s_0}\\
	=&\frac{\sqrt{\lambda}u(p_0,z_0)}{1+a \sqrt{\lambda}u(p_0,z_0)}
	\geq \frac{\sqrt{\lambda}C_2e^{-\sqrt{\lambda} z_0}}{1+a C_0 e^{-\sqrt{\lambda} z_0}}\geq\frac{\sqrt{\lambda}C_2}{2}e^{-\sqrt{\lambda} z_0}
	\end{align*}

	The last inequality holds when $ z_0 $ is large enough.	
\end{proof}

We denote $\Sigma^u_{z}:=u(\Sigma,z )$. Suppose $ V $ is the open domain enclosed by $ \Sigma^u_{\bar{ z}}  $ and $ \Sigma $. Then  $\Sigma^u_{z}$ will stay in the tubular area $ V\times\R $ if $  z >\bar{ z}$.

Denote $ N=\text{graph} f =\text{graph} u $. The Harnack inequality in \cite[Proposition 2]{yau} implies that there are constants $\rho, c_2$,  such that  for any point $ X_0 $ on $ N $, $$\sup_{N\cap B^4_\rho(X_0)}<\ov{e}_4,\nu>\leq c_2 \inf_{N\cap B^4_\rho(X_0)}<\ov{e}_4,\nu>$$

where $$<\ov{e}_4,\nu>=\frac{1}{\sqrt{1+|\nabla f|^2}}=\frac{-\partial_z u}{\sqrt{1+|^{C_{\bar z}}\nabla u|^2}}$$

Now  $ V $ can be covered by finite balls $ B^4_\rho $, and we suppose constant $ c_4 $ is the number of balls that is enough  to cover $ V $. Thus each $\Sigma^u_{z}$  can be covered by $ c_4 $ balls $ B^4_\rho $, for $\forall z >\bar{ z}$. This is because each of them stays in a horizontal cut of the tube  $ V\times\R $, which is exactly $ V $.

Therefore, together with the Harnack inequality,  there is a constant $c_5$, such that 
$$\sup_{\Sigma^u_{z }}<\ov{e}_4,\nu>\leq c_5 \inf_{\Sigma^u_{z }  }<\ov{e}_4,\nu>,\mbox{ } \forall z>\bar{ z}.$$
Now for any $p\in\Sigma$, we have 
{\allowdisplaybreaks\begin{align*}
	|\partial_z u(p,z)|\geq& \frac{|\partial_z u(p,z)|}{\sqrt{1+|^{C_{\bar z}}\nabla u(p,z)|^2}}=\frac{-\partial_z u(p,z)}{\sqrt{1+|^{C_{\bar z}}\nabla u(p,z)|^2}}\\
	=&<\ov{e}_4,\nu>\big|_{(p,z)}\geq \inf_{\Sigma^u_{z}}<\ov{e}_4,\nu>\\
	\geq& c_5^{-1}\sup_{\Sigma^u_{z}}<\ov{e}_4,\nu>  \geq c_5^{-1}<\ov{e}_4,\nu>\big|_{(p_0,z)}\\
	=& \frac{-\partial_z u(p_z,z)}{c_5\sqrt{1+|^{C_{\bar z}}\nabla u(p_z,z)|^2}}\geq \frac{ C e^{-\sqrt{\lambda}z}  }{c_5\sqrt{1+C_1^2 e^{-2\sqrt{\lambda}z}}}\\
	\geq&\frac{C}{2c_5 C_1}e^{-\sqrt{\lambda }z}
	\end{align*}}
The last line holds if $ z $ is big enough.\par 
Set $ C_3= \frac{C}{2c_5 C_1} $ then we prove the following:

\begin{proposition}
	Assume the condition and notation of Theorem \ref{ex}, then there exist  constants $ C_3>0 $, $z_2\geq \bar{z}$,  such that $ \forall z\geq z_2, p\in\Sigma, $
	$$|\partial_z u(p,z)|\geq C_3 e^{-\sqrt{\lambda}z}$$
	
\end{proposition}

Now, by Implicit Function Theorem we have:
\begin{equation*}
|\partial_s f(p,s)|=  \frac{1}{|\partial_z u(p,z)|}
\leq  \frac{1}{C_3 e^{-\sqrt{\lambda} z}}\leq \frac{C_1}{C_3  s}
\end{equation*}
Moreover,  
\begin{equation*}
|\partial_i f(p,s)|=|\partial_i u(p,z)||\partial_s f(p,s)|
\leq  \frac{C_1^2 e^{-\sqrt{\lambda} z}}{C_3 s}\leq \frac{C_1^2}{C_2 C_3}
\end{equation*}
Because $ \partial_i, i=1,2 $ are the derivatives taken in the direction  of the orthonormal frame $ e_i, i=1,2 $  on $ \Sigma_s $, we have:
\begin{equation*}
|\nabla^{\Sigma_s} f(p,s)|\leq \frac{\sqrt{2} C_1^2}{C_2 C_3}
\end{equation*}
holds for $ \forall p\in\Sigma$ , $\forall s\in(0, s_0] $. Here $ \nabla^{\Sigma_s}$ denotes the covariant derivative along $ \Sigma_s. $

Put together all the above conclusions, we prove the following:
\begin{proposition}
	
	Under the same assumptions as  Theorem \ref{thm1}, suppose $f$ is a blowup solution of Eq.\eqref{jang} in an open neighborhood $ U $ near $ \Sigma $, which is a compact boundary component and a strictly stable MOTS with principal eigenvalue $ \lambda>0 $. Denote $ N=\f $.
	Then,   under coordinate system $ \bar{\Psi} $, \par
	(1) there exist positive constants $\bar z $  and $C_1 $, $C_2 $, $C_3 $, such that $N \cap( U
	\times [\bar z,\infty))$ can be written as the graph of a function
	$u$ over $C_{\bar z}:= \Sigma \times [\bar z,\infty)$,  and
	\begin{gather} 
	|u(p,z)| + |^{C_{\bar z}}\nabla u(p,z)| + |^{C_{\bar z}}\nabla^2 u(p,z)| \leq C_1\exp(-\sqrt{\lambda} z)\\
	|u(p,z)|\geq	C_2\exp(-\sqrt{\lambda} z) \\
	|^{C_{\bar z}}\nabla u(p,z)|\geq	C_3\exp(-\sqrt{\lambda} z) 
	\end{gather} 
	where $^{C_{\bar z}}\nabla$ is the covariant derivative w.r.t. the induced metric on $C_{\bar{z}}$.\par 
	(2) Denote $ \Sigma_s=\Psi_B(\Sigma,s) $. Then there exists constant $ s_0 $, such that the following gradient estimates for $ f $:
	\begin{align}\label{cy barrier3}
	\frac{C_2}{C_1s }\leq	|\partial_s  f (p,s)  |\leq\frac{C_1}{C_3s }\\
	|\nabla^{\Sigma_s} f(p,s)|\leq \frac{\sqrt{2} C_1^2}{C_2 C_3}
	\end{align}
	hold for $ \forall p\in\Sigma$ , $\forall s\in(0, s_0] $. Here $ \nabla^{\Sigma_s}$ denotes the covariant derivative along $ \Sigma_s. $
\end{proposition}

Thus Theorem \ref{thm2} gets proved.

\section{Dependence of the Coefficients}\label{apriori}

In this section we   prove Theorem \ref{thm3}. We prove that for the solution $ f_0 $ constructed by J. Metzger, c.f. \cite{metzger}, we can have the same estimates as in Theorem \ref{thm2} hold for $ f_0 $, while with constants $ \bar{ z}, C_1, C_2, C_3 $  depending only on the geometry of the initial data set.\par

To remove the dependence of $ \bar{ z}, C_1, C_2, C_3 $ on the vertical translation on the solution $ f $, we need to normalize $ f $ such that there is no such kind of freedom. For example, we can normalize $ f $ by some outer boundary condition or require that the solution decay at infinity. J. Metzger constructed a blowup solution $ f_0  $ in \cite{metzger} with these constraints, thus in this section we will study this specific example  $ f_0  $  to see if our $ C_1,C_2,C_3,\bar{ z} $ can be determined by the geometry of initial data. Therefore, although Theorem \ref{thm1} and   Theorem \ref{thm2} work for more general initial data sets, we need to further assume that $ \M $ is 3-dim asymptotically flat manifold with one end, satisfies the dominant energy condition. Also assume  $ \Sigma $ is the only boundary component, and is a compact strictly stable outermost MOTS. \par
Then by \cite[Lemma 9.7]{stable}, $ \Sigma $ is topologically 2-sphere.
We also assume  there is no   MITS in $ \M $. The existence of asymptotically decaying solution of Jang's equation over $ \M $, which only blows up at $ \Sigma $ then follows \cite[Theorem  3.1, Remark 3.3]{metzger} . The detailed construction procedure can be found in \cite[Section 3]{metzger}. We only give a brief introduction about the procedure here for consistence. \par 

First of all, it can be  constructed  from $ (\M,g,k) $ a new data set  $(\tilde M, \tilde g, \tilde k)$ (c.f. \cite{metzger}, \cite{Andersson}), with following properties:
\begin{enumerate}
	\item $\M\subset \tilde \M$ with $\tilde g|_\M = g$, $\tilde k|_\M = k$,
	
	\item $\theta^+[\partial \tilde \M] <0$,
	\item $H[\partial\tilde \M] > 0$ where $H$ is the mean curvature of
	$\partial\tilde{\M}$ with respect to the normal pointing out of $ \tilde{\M} $,
	\item The region $\tilde \M \setminus \M$ is foliated by surfaces
	$\Sigma_s$ with $\theta^+(\Sigma_s) <0$.
\end{enumerate}

Then the following Dirichlet boundary value problem can be solved:
\begin{equation}
\label{eq3}
\begin{cases}
\J[f_t] = t f_t & \text{in}\ \tilde M \\
f_t      = \frac{\delta}{2t} & \text{on}\ \partial^-\tilde M
\\
f_t     \to 0 & \mbox{when } x\to \infty
\end{cases}    
\end{equation}
where $\delta$ is a lower bound for the mean curvature $H$ on $\partial\tilde{\M}$. 
The solution $f_t$ of Eq.\eqref{eq3} satisfies an estimate
of the form
\begin{equation}
\label{eq:4}
\sup_{\tilde \M} |f_t| + \sup_{\tilde \M} |\nabla f_t| \leq \frac{C}{t},
\end{equation}
where $C$ is a constant depending only on the data $(\tilde M,\tilde
g, \tilde k)$ but not on $t$.

This gradient estimate and the fact that graphs $N_t=\f_t$ have uniformly bounded curvature in $\tilde \M\times \R$ away from the boundary implies that it allows to extract a sequence $t_i\to 0$ such that the $N_{t_i}$ converge smoothly to a
manifold $N_0$ (c.f. \cite[Section 4]{yau}, \cite[Proposition 3.8]{Andersson}), which can be proved to be the graph of a function $ f_0 $ on $ \M $ which satisfies
$\J[f_0]=0$, with the desired asymptotics. \par 

From the local parametric estimate \cite[Proposition 2]{yau}, the a priori bounds at asymptotically flat ends \cite[Proposition 3]{yau}, and the gradient estimate \eqref{eq:4}, it is straightforward to see that for any $ \M'\subset\M $ such that $ \partial \M' $ does not intersect with $ \Sigma $,
there is a uniform $ C^0 $ bound for $ f_t $ on $ \M' $ $$  |\sup_{\M'} f_t|<c_1.$$ where $ c_1 $ only depends  on the geometry of the initial data. This also bounds the limit function $ f_0 $. Thus by setting $ \M' $ to be the portion of $ \M $ outside $ \Sigma_{s_0} $ in Theorem \ref{thm1}, we have  $-c_1< \inf_{q\in\Sigma}f_0(q,s_0)\leq\sup_{q\in\Sigma}f_0(q,s_0)< c_1 $ in  Eq.\eqref{thm barrier}. Thus we have the following:
\begin{proposition}\label{thm3 half}
	
	Besides the conditions and notations in Theorem \ref{thm2},  we further assume that $ \M $ is a 3-dim asymptotically flat manifold with one end, and satisfies the dominant energy condition. Also assume that $ \Sigma $ is the only boundary component, and is a compact outermost MOTS. We also assume that there is no MITS in $ \M $. Then a function $ f_0 $ on $ \M $ can be constructed as in \cite[Theorem 3.1]{metzger}, such that $ \J[f_0]=0 $, $ f_0(x)\to 0 $ when $ |x|\to \infty $, and it only blows up at $ \Sigma $. Denote $ N_0=\f_0 $. Then, 
	
	(1) there exist positive constants $\bar z $  and $C_1 $, $C_2 $, $C_3$, which only depend  on the initial data  such that $N_0 \cap( U
	\times [\bar z,\infty))$ can be written as the graph of a function
	$u_0$ over $C_{\bar z}:= \Sigma \times [\bar z,\infty)$ under coordinate system $ \bar{\Psi} $,  and
	\begin{gather} 
	|u_0(p,z)| + |^{C_{\bar z}}\nabla u_0(p,z)| + |^{C_{\bar z}}\nabla^2 u_0(p,z)| \leq C_1\exp(-\sqrt{\lambda} z).\\
	|u_0(p,z)|\geq	C_2\exp(-\sqrt{\lambda} z) \\
	|^{C_{\bar z}}\nabla u_0(p,z)|\geq	C_3\exp(-\sqrt{\lambda} z) 
	\end{gather}
	where $^{C_{\bar z}}\nabla$ is the covariant derivative w.r.t. the induced metric on $C_{\bar{z}}$.\\ 
	
	(2) Denote $ \Sigma_s=\Psi_B(\Sigma,s) $. Then
	\begin{align}\label{cy barrier4}
	\frac{C_2}{C_1s } \leq	|\partial_s  f_0 (p,s)  |\leq\frac{C_1}{C_3s }\\
	|\nabla^{\Sigma_s} f_0(p,s)|\leq \frac{\sqrt{2} C_1^2}{C_2 C_3}
	\end{align}
	hold for $ \forall p\in\Sigma$ , $\forall s\in(0, s_0] $. Here $ \nabla^{\Sigma_s}$ denotes the covariant derivative along $ \Sigma_s. $
\end{proposition}

Thus Theorem \ref{thm3} is proved.

\section{Application to Penrose Inequality}\label{penrose}

In previous sections we have obtained a sharp estimate for the blowup rate and gradient  of the solution $ f_0 $ of Jang's equation, which is contructed by J. Metzger in \cite{metzger}. In this section we apply this control to the slice $ \f_0 $ get a   Penrose-like inequality for general initial data sets.

\subsection{Basic Settings}
Assume the same condition as in Theorem \ref{thm3}. We will apply our estimates onto the solution  $ f_0 $ constructed in \cite[Theorem 3.1]{metzger} . Also in this section we still use   Foliation B.  Recall that in this foliation $ \Psi_B: \Sigma\times[0,\bar{ s}]\to \M $ the metric $ g $ of $ \M $ near a neighborhood of $ \Sigma $ can be written as 
\begin{equation}
g=\beta^2 ds^2+g_{\Sigma_s}
\end{equation}
and $ \theta^+  $ satisfies:
\begin{equation}
\theta^+[\Sigma_s]=\lambda \beta s +O(s^2)
\end{equation}

Now by  Theorem \ref{thm3} we know there are constants  $ C_1, C_2,C_3  $ which only depend on the initial data, such that: 
\begin{align} 
\frac{C_2}{C_1 s}\leq|\partial_s  f_0   |\leq\frac{C_1}{C_3 s}\\
|\nabla^{\Sigma_s} f_0(p,s)|\leq \frac{\sqrt{2} C_1^2}{C_2 C_3}
\end{align}
Here $ \nabla^{\Sigma_s}$ denotes the covariant derivative along $ \Sigma_s. $\par 

In this Section we will use $ \ov{\M} $ to represent  $ N_0=\f_0 $. 
Denote the induced metric on Jang's slice to be $ \ov{g} $. Denote $ \nabla,\ov{\nabla} $ to be the Levi-Civita connection of $ g,\ov{g} $, respectively. Then $ \ov{g}=df_0^2+g $, and the second fundamental form on $ \ov{\M} $ embedded into $ \M\times\R $ will be $ h=\dfrac{\nabla^2 f_0}{\sqrt{1+|\nabla f_0|^2}} $. Denote  $ \ov{\Sigma}_s $ to be the lift of $ \Sigma_s $   to $ \ov{\M} $. Denote the portion of $ \M $, $ \ov{\M} $ outside $ \Sigma_s $, $ \ov{\Sigma}_s $ to be $ \M_s $, $ \ov{\M}_s $, respectively.  We choose orthonormal frames $ \{e_1,e_2\} $ and $ \{ \ov{e}_1,\ov{e}_2   \} $ for $ T\Sigma_s $ and $ T\ov{\Sigma}_s $, respectively. Let $ e_3  $ be the normal of $ \Sigma_s $ pointing into $ \M_s $ that is tangent to $ \M $. We also choose  $ \{ \ov{e}_3,\ov{e}_4   \} $ for the normal bundle of $ \Sigma_s $ in $ \M\times\R $, such that $ \ov{e}_3 $ is tangent to the graph $ \ov{\M} $ and $ \ov{e}_4 $ is a downward unit normal vector of $ \ov{\M} $ in $ \M\times\R $. We denote the induced Levi-Civita connection on $  \Sigma _s $ to be $  \slashed{\nabla} $ and $ \phi_s=f_0|_{  \Sigma_s} $  to be $ f $'s restriction on $   \Sigma_s $. Denote $ H_{\Sigma_s}  $, $ \ov{H}_{\ov{\Sigma}_s}  $ as the mean curvature of $\Sigma_s $, $\ov{\Sigma}_s $ with respect to $ e_3 $, $ \ov{e}_3 $, respectively.  Recall that the scalar curvature on the Jang's slice can be written as:
\begin{equation}
\ov{R}=16 \pi(\mu-J(\omega))+|h-k|^2_{\ov{g}}+2|q|^2_{\ov{g}}-2div_{\ov{g}}(q)
\end{equation}
where 
\begin{align*}
\omega_i=&\frac{\nabla_i f_0}{\sqrt{1+|\nabla f_0|^2}}\\
q_i=&\frac{f_0^j}{\sqrt{1+|\nabla f_0|^2}} (h_{ij}-k_{ij})
\end{align*}

We will calculate  $ \ov{H}_{\ov{\Sigma}_s}-q(\ov{e}_3) $ in terms of the geometry quantities on $ \M $.
\subsection{Calculation of $ \ov{H}_{\ov{\Sigma}_s}-q(\ov{e}_3) $}

We now calculate the value of $ \ov{H}_{\ov{\Sigma}_s}-q(\ov{e}_3) $ in terms of the geometry quantities on $ \M $.

From \cite[page 10]{black}:
\begin{equation}
\ov{H}_{\ov{\Sigma}_s}-q(\ov{e}_3)=(<e_3,\ov{e}_3>+\frac{<e_3,\ov{e}_4>^2}{<e_3,\ov{e}_3>})H_{\Sigma_s} -\frac{<e_3,\ov{e}_4>}{<e_3,\ov{e}_3>} Tr_{\Sigma_s}k
\end{equation}

Following the calculation of  \cite[Section 4, Eq.(4.4)]{isometric}, we get:
\begin{equation}
<e_3,\ov{e}_3>=\frac{\sqrt{1+| \slashed{\nabla} \phi_s|^2}}{\sqrt{1+| \nabla f_0|^2}} \mbox{ }\mbox{and }\mbox{ }<e_3,\ov{e}_4>=\frac{ \nabla _{e_3}f_0}{\sqrt{1+| \nabla f_0|^2}}
\end{equation}
Then the above equations imply:
\begin{equation}\label{sy}
\ov{H}_{\ov{\Sigma}_s}-q(\ov{e}_3)=\frac{\sqrt{1+| \nabla f_0|^2}}{\sqrt{1+| \slashed{\nabla} \phi_s|^2}}\theta^+[\Sigma_s] +\frac{\sqrt{1+| \slashed{\nabla} \phi_s|^2}}{| \nabla _{e_3}f_0|+\sqrt{1+| \nabla f_0|^2}} Tr_{\Sigma_s}k
\end{equation}
In the above calculation we use the  fact that $  \nabla _{e_3}f_0<0 $ and $ | \nabla f_0|^2=|\nabla _{e_3}f_0|^2+| \slashed{\nabla} \phi_s|^2 $.

Therefore,  we can have an upper bound for $ \ov{H}_{\ov{\Sigma}_s}-q(\ov{e}_3) $ as follows:

\begin{align}
\ov{H}_{\ov{\Sigma}_s}-q(\ov{e}_3)&=\frac{\sqrt{1+| \nabla f_0|^2}}{\sqrt{1+| \slashed{\nabla} \phi_s|^2}}\theta^+[\Sigma_s] +\frac{\sqrt{1+| \slashed{\nabla} \phi_s|^2}}{| \nabla _{e_3}f_0|+\sqrt{1+| \nabla f_0|^2}} Tr_{\Sigma_s}k\\\nonumber
&\leq \sqrt{1+| \nabla f_0|^2}(\lambda \beta s+O(s^2))+O(s)\\\nonumber
&=\sqrt{\beta^{-2}|\partial_s f_0|^2+O(1)}(\lambda \beta s+O(s^2))+O(s)\\\nonumber
&\leq\sqrt{ \frac{C_1^2}{\beta^2 C_3^2s^2}+O(1)}\lambda \beta s+O(s)\\\nonumber
&=\lambda \frac{C_1}{C_3} +O(s)\label{sy22}
\end{align}

\begin{remark}
	The standard gradient estimate for the mean curvature type equation actually provides an estimate for $ |\nabla f_0| $.  From \cite[Chapter 16]{Gilbarg}, there exist positive constants $ C_4, C_5 $ only depends on $ \sup |h| $ and $ \sup |\nabla h| $, where $ h $  is the mean curvature of $ N_0=\f_0 $ embedded into $ \M\times\R $, such that:
	$$ |\nabla f_0|\leq C_4 \exp({C_5 f_0/\tau}) $$
	Then because $ f=-\frac{1}{\sqrt{\lambda}}\log\tau +O(1)$ in a neighborhood of $ \Sigma $,
	there exists positive constant  $ C_6 $, s.t.
	$$ |\nabla f_0|\leq C_4 \exp({C_6 \frac{\log\frac{1}{\tau}}{\tau}}) $$
	This is an upper bound for the gradient estimate of $ f_0 $. But the order is not sharp enough for proving a Penrose-like inequality in this cection.
\end{remark}

\subsection{Existence of Harmonic Spinor}

We are now going to find a harmonic spinor on $ \ov{\M} $ which is constant at infinity. Because $ \ov{\M} $ is 3-dimensional oriented Riemannian manifold, it is also a spin manifold.  
Denote $\ov{\M}_r $ as the domain enclosed by large sphere $ S_r $, and $ \ov{\M}_{s,r}=\ov{\M}_r \cap \ov{\M}_s $.  Denote $ \nu_r  $ to be the outer normal of $ S_r $.\par
Our calculation will follow \cite[Section 2]{herz} and  \cite[Section 5]{isometric}. 
From \cite[Lemma 2.3]{herz},  we have the following Bochner-Lichnerowicz-Weitzenb\"{o}ck formula:
\begin{align}
\int_{\ov{\M}_{s,r}}|D\psi|^2
&=\int_{\ov{\M}_{s,r}}|\ov{\nabla}\psi|^2
+\frac{1}{4}\int_{\ov{\M}_{s,r}} \ov{R}|\psi|^2
+\int_{\partial\ov{\M}_{s}} <\ov{\nabla}_{\ov{e}_3} \psi+\ov{e}_3\cdot D\psi,\psi>
\\\nonumber& -\int_{S_r} <\ov{\nabla}_{\nu_r} \psi+\nu_r\cdot D\psi,\psi>\label{bochner}
\end{align}
where $ \psi $ is a spinor on $ \ov{\M}_{s,r} $,   $ D $ is the Dirac operator and $ \cdot $ is Clifford multiplication on Jang's slice $ \ov{\M}_{s,r} $. Here we still use $ \ov{\nabla} $ to denote the spin connection on $ \ov{\M}_{s,r} $

From \cite[Lemma 9.7]{stable} we know that $$ \lambda|\partial\M|_g\leq 4\pi,$$

We set $$ C_s=\min_{\partial\ov{\M}_s} \left(4-(\ov{H}_{\partial\ov{\M}_s}-q(\ov{e}_3))\sqrt{\dfrac{|\partial\ov{\M}_s|_{\ov{g}}}{\pi}}\right). $$ 

Then 
\begin{align}
C_s&\geq 4-(\ov{H}_{\partial\ov{\M}_s}-q(\ov{e}_3))\sqrt{\dfrac{|\partial\ov{\M}_s|_{\ov{g}}}{\pi}}\\\nonumber
&\geq 4- (\lambda \frac{C_1}{C_3} +O(s)) \sqrt{\dfrac{|\partial\ov{\M}_s|_{\ov{g}}}{|\partial{\M}_s|_{{g}}}}\sqrt{\dfrac{|\partial{\M}_s|_{{g}}}{\pi}}\\\nonumber
&\geq  4- (\lambda \frac{C_1}{C_3} +O(s)) \sup_{\Sigma_{s }}\sqrt{1+| \slashed{\nabla} \phi_s|^2}\sqrt{\dfrac{4}{\lambda}}\\\nonumber
&=  4-2\sqrt{\lambda}(\frac{C_1}{C_3} +O(s))\sup_{\Sigma_{s }}\sqrt{1+|  \nabla^{\Sigma_s}  f|^2} \\\nonumber
&\geq 4-2\sqrt{\lambda}(\frac{C_1}{C_3} +O(s))\sqrt{1+\frac{2C_1^4 }{C_2^2C_3^2}}
\end{align}
If we suppose 
$$\lambda \frac{C_1^2}{C_3^2}(1+\frac{2 C_1^4}{C_2^2C_3^2})<4, $$
then $ C_s>0 $ if $ s $ is small enough.  Therefore, there exists a positive constant, still denote as  $ s_0  $, s.t. $ \forall s\in(0,s_0] $, we have $ C_s>0 $. In all the context below,   $ s_0 $ always refers to this definition.\par

We are now going to find a nontrivial solution to the PDE system as  in \cite{herz}:
\begin{equation}
D\psi=0 \text{ }\text{ with }\text{ }\text{boundary }\text{ }\text{condition }\text{ }\text{ }  P^+_s\psi=0 
\end{equation}
where $ \psi  $ is asymptotically constant and $ P^+_s $ is the $ L^2  $ orthogonal projection on the space of eigenvectors of positive eigenvalues of $ \slashed{D} $ on $ \partial\ov{\M}_s $. Here $ \slashed{D} $  
is the 2-dimensional Dirac operator on the boundary, defined as:
\begin{equation*}
-\slashed{D}\psi=\ov{e}_3\cdot\ov{e}_1\cdot \ov{\slashed{\nabla}}_{\ov{e}_1}\psi+\ov{e}_3\cdot\ov{e}_2\cdot \ov{\slashed{\nabla}}_{\ov{e}_2}\psi
\end{equation*}
where $ \ov{\slashed{\nabla}} $ is the spin connection on the boundary. 

Notice that the boundary is a topological  2-sphere when s is small, there is no $ \slashed{D} $- harmonic spinor field on it. We denote the
weighted Sobolev space
\begin{equation*}
W^{1,2}_{-1}(\ov{\M}_{s})=\{v\in
W^{1,2}_{loc}(\ov{\M}_{s})\mid
|x|^{l-1}\ov{\nabla}^{l}v\in L^{2}(\ov{\M}_{s}),\text{
}\text{ }l=0,1\}.
\end{equation*} 
as a  notation for both functions and spinors. We then define the space 
\begin{equation}
\HH_s=\{ \psi\in W_{-1}^{1,2}(\ov{\M}_s), P^+_s\psi=0  \}
\end{equation}  
This is a Hilbert space w.r.t the $W_{-1}^{1,2}  $ norm.\par

The boundary term for $\partial\ov{\M}_s  $ in the Weitzenb\"{o}ck formula can be rearranged as
\begin{equation}
\int_{\partial\ov{\M}_{s}} <\ov{\nabla}_{\ov{e}_3} \psi+\ov{e}_3\cdot D\psi,\psi>=-\int_{\partial\ov{\M}_s}<\psi,\slashed{D}\psi>-\frac{1}{2}\int_{\partial\ov{\M}_s}\ov{H}_{\partial\ov{\M}_s}|\psi|^2
\end{equation}

We will follow the steps of calculation in \cite[Theorem 5.1]{isometric} . First for any domain $ \Omega $ and vector field $ X $,  the following equation holds:
\begin{equation*}
\int_{\partial\Omega}<X,\nu_{out}>|\psi|^2
=\int_{\Omega} divX|\psi|^2+\int_\Omega X(|\psi|^2)
\end{equation*}

By setting $ X=q,\Omega= \ov{\M}_{s,r}$ in the above equation, the Weitzenb\"{o}ck formula can be rearranged as:
\begin{align}\label{ddd}
\int_{\ov{\M}_{s,r}}|D\psi|^2
=&\int_{\ov{\M}_{s,r}}|\ov{\nabla}\psi|^2
+\frac{1}{4}\int_{\ov{\M}_{s,r}} (\ov{R}+2div_{\ov{g}}(q))|\psi|^2
+\frac{1}{2}\int_{\ov{\M}_{s,r}} q(|\psi|^2)\\
\nonumber
&+\frac{1}{2}\int_{\partial\ov{\M}_{s}} q(\ov{e}_3)|\psi|^2
-\frac{1}{2}\int_{S_r} q(\nu_r)|\psi|^2\\
\nonumber 
&-\int_{\partial\ov{\M}_{s}}<\psi,\slashed{D}\psi>
-\frac{1}{2}\int_{\partial\ov{\M}_{s}}\ov{H}|\psi|^2
-\int_{S_r} <\ov{\nabla}_{\nu_r} \psi+{\nu_r}\cdot D\psi,\psi>\\
\nonumber 
\geq& \int_{\ov{\M}_{s,r}}|\ov{\nabla}\psi|^2
+\frac{1}{2}\int_{\ov{\M}_{s,r}}(|q|^2_{\ov{g}} |\psi|^2+ q(|\psi|^2))\\
\nonumber
&+\int_{\partial\ov{\M}_{s}}-<\psi,\slashed{D}\psi>-\frac{1}{2}\ov{H}|\psi|^2+\frac{1}{2} q(\ov{e}_3)|\psi|^2\\
\nonumber
&-\int_{S_r} \frac{1}{2}q(\nu_r)|\psi|^2+<\ov{\nabla}_{\nu_r} \psi+{\nu_r}\cdot D\psi,\psi>	 
\end{align}

The second inequality of the above is induced by Eq.\eqref{syid}
\begin{equation}
\ov{R}-2div_{\ov{g}}(q)\geq 2|q|^2_{\ov{g}}
\end{equation}
Now if we set  $ X=q $  in the following inequality (c.f. \cite{isometric}):
\begin{equation}\label{isometric}
|\ov{\nabla}\psi|^2+|X|^2|\psi|^2+X(|\psi|^2)\ge 0
\end{equation}
as 
\begin{equation*}
X(|\psi|^2) 
=2<\ov{\nabla}_X\psi,\psi> 
\geq 	 -2|\ov{\nabla}_X\psi|  |\psi| 
\geq  	 -2|\ov{\nabla}\psi||X| |\psi| 
\end{equation*}
Then Ineq.\eqref{ddd} will imply:
\begin{align}\label{bochner2}
\int_{\ov{\M}_{s,r}}|D\psi|^2&\geq 
\frac{1}{2}\int_{\ov{\M}_{s,r}}|\ov{\nabla}\psi|^2\\
\nonumber
&  +\int_{\partial\ov{\M}_{s}}-<\psi,\slashed{D}\psi>-\frac{1}{2}\ov{H}|\psi|^2+\frac{1}{2} q(\ov{e}_3)|\psi|^2\\
\nonumber
&   -\int_{S_r} \frac{1}{2}q(\nu_r)|\psi|^2+<\ov{\nabla}_{\nu_r} \psi+{\nu_r}\cdot D\psi,\psi>
\end{align}

For any spinor $ \psi\in W_{-1}^{1,2} $, the Ineq.\eqref{bochner2} becomes (because $ C_c^\infty $ in dense in $  W_{-1}^{1,2} $):
\begin{align}\label{bochner3}
\int_{\ov{\M}_{s,r}}|D\psi|^2
&\geq\frac{1}{2}\int_{\ov{\M}_{s,r}}|\ov{\nabla}\psi|^2
+\int_{\partial\ov{\M}_{s}}-<\psi,\slashed{D}\psi>-\frac{1}{2}(\ov{H}-q(\ov{e}_3))|\psi|^2
\end{align}

Now for spinors $ \psi  $ in $ \HH_s $, we decompose $ \psi  $ in the eigenspace of $ \slashed{D} $ as: $ \psi=\sum_i\psi_i $, where $ \psi_i $ is the eigenfunctions w.r.t to $ \alpha_i $, i.e. $ \slashed{D}\psi_i=\alpha_i\psi_i $. 
Under this decomposition, the boundary term of formula Ineq.\eqref{bochner3} becomes:
\begin{align*}
& \int_{\partial\ov{\M}_s}-<\psi,\slashed{D}\psi> -\frac{1}{2}\ov{H}|\psi|^2+\frac{1}{2}q(\ov{e}_3)|\psi|^2 \\
&=  \int_{\partial\ov{\M}_s} -\sum_i \alpha_i|\psi_i|^2 -\frac{1}{2}\ov{H}|\psi|^2+\frac{1}{2}q(\ov{e}_3)\sum_i|\psi_i|^2   \\
&=\sum_i \int_{\partial\ov{\M}_s} (-\alpha_i-\frac{1}{2}\ov{H}+\frac{1}{2}q(\ov{e}_3))|\psi_i|^2
\end{align*}

From the lower estimation for the absolute value of  eigenvalues of Dirac operator on 2-sphere \cite{bar},\cite{hi}, and together with the fact that our boundary condition $ P^+_s\psi=0 $ forbids any positive eigenvalues, we have $ \alpha_i\le-2\sqrt{\dfrac{\pi}{|\partial\ov{\M}_s|_{\ov{g}}    }} $. 
Now for $ s<s_0 $, we have:
\begin{equation}
-\alpha_i-\frac{1}{2}\ov{H}+\frac{1}{2}q(\ov{e}_3)
\geq
2\sqrt{\dfrac{\pi}{|\partial\ov{\M}_s|_{\ov{g}}    }}-\frac{1}{2}\ov{H}+\frac{1}{2}q(\ov{e}_3)
\geq
\frac{1}{2}\sqrt{\dfrac{\pi}{|\partial\ov{\M}_s|_{\ov{g}}    }} C_s	>0
\end{equation} 

Thus
\begin{align}
& \int_{\partial\ov{M}}-<\psi,\slashed{D}\psi> -\frac{1}{2}\ov{H}|\psi|^2+\frac{1}{2}q(\ov{e}_3)|\psi|^2 \\
\nonumber
=&\sum_i \int_{\partial\ov{M}_s} (-\alpha_i-\frac{1}{2}\ov{H}+\frac{1}{2}q(\ov{e}_3))|\psi_i|^2\\
\nonumber
\geq&\sum_i \int_{\partial\ov{M}_s} 	\frac{1}{2}\sqrt{\dfrac{\pi}{|\partial\ov{\M}_s|_{\ov{g}}    }} C_s|\psi_i|^2\\
\nonumber
=&	\frac{1}{2}\sqrt{\dfrac{\pi}{|\partial\ov{\M}_s|_{\ov{g}}    }} C_s \int_{\partial\ov{M_s}}|\psi|^2\\
\nonumber
\geq& 0
\end{align}
This formula together with formula \ref{bochner3} implies the following estimate for spinors in $ \HH_s $:
\begin{equation}\label{bochner4}
\int_{\ov{\M}_{s}}|D\psi|^2\geq	\frac{1}{2}\int_{\ov{\M}_{s}}|\ov{\nabla}\psi|^2+\frac{1}{2}\sqrt{\dfrac{\pi}{|\partial\ov{\M}_s|_{\ov{g}}    }} C_s \int_{\partial\ov{M_s}}|\psi|^2\geq\frac{1}{2}\int_{\ov{\M}_{s}}|\ov{\nabla}\psi|^2
\end{equation}

The right hand side of the above is the Hilbert norm on $ W_{-1}^{1,2} $. This implies a coercive estimate for spinors in $ \HH $, which is enough to establish the existence of harmonic spinors that are asymptotic to constant spinor at infinity. We are then able to get a lower bound for the  ADM mass of the initial data set: \par

\begin{proposition}\label{ppmt}
	Under Foliation B, we suppose 
	$$\lambda \frac{C_1^2}{C_3^2}(1+\frac{2 C_1^4}{C_2^2C_3^2})<4, $$ where $ C_1, C_2, C_3 $ are constants appearing in the gradient estimates in Theorem \ref{thm3}. Also we suppose $ \psi_0 $ is  a smooth spinor field  which is constant in some chart around infinity and $ P^+_s\psi_0=0 $.  Then for $ s<s_0 $, there exists a unique $ \psi\in \HH_s $ s.t. 
	\begin{equation}\label{spin pde}
	D(\psi_0+\psi)=0, P^+_s(\psi_0+\psi)=0
	\end{equation}
	Furthermore, we have the following lower bound for the  ADM mass of the initial data set:
	\begin{equation}\label{ppmt eq}
	8\pi| \psi_0|^2 m\geq  \int_{\ov{\M}_{s}}|\ov{\nabla}(\psi+\psi_0)|^2+ \sqrt{\dfrac{\pi}{|\partial\ov{\M}_s|_{\ov{g}}    }} C_s \int_{\partial\ov{M_s}}|\psi+\psi_0|^2
	\end{equation}
\end{proposition} 
\begin{proof}
	We have already establish the existence of the solution of Eq.\eqref{spin pde}. Plugging  $ \psi+\psi_0 $ back into Ineq.\eqref{bochner2}  will imply the following   result:
	\begin{align}\label{pen}
	& \int_{S_r}\frac{1}{2}q(\nu_r)|\psi_0+\psi|^2
	+<\ov{\nabla}_{\nu_r} (\psi+\psi_0)+\nu_r\cdot D (\psi+\psi_0), \psi+\psi_0>\\
	\nonumber
	\geq&
	\frac{1}{2}\int_{\ov{\M}_{s}}|\ov{\nabla}(\psi+\psi_0)|^2+\frac{1}{2}\sqrt{\dfrac{\pi}{|\partial\ov{\M}_s|_{\ov{g}}    }} C_s \int_{\partial\ov{M_s}}|\psi+\psi_0|^2
	\end{align}
	We take limit $ r\to\infty $ in this inequality. The decay rate of $ q $ implies that the first term of left hand side disappears under this limiting process. Together with the fact that 
	\begin{equation}
	4\pi| \psi_0|^2 m=\lim_{r\to\infty}\int_{S_r}<\ov{\nabla}_{\nu_r} (\psi+\psi_0)+\nu_r\cdot D (\psi+\psi_0), \psi+\psi_0>
	\end{equation}
	We then prove the following lower bound for the ADM mass:
	\begin{equation}
	8\pi| \psi_0|^2 m\geq  \int_{\ov{\M}_{s}}|\ov{\nabla}(\psi+\psi_0)|^2+ \sqrt{\dfrac{\pi}{|\partial\ov{\M}_s|_{\ov{g}}    }} C_s \int_{\partial\ov{M_s}}|\psi+\psi_0|^2
	\end{equation}	
\end{proof}

\subsection{Discussion about Condition 2}
In this section we prove that the similar conclusion still holds under Foliation C if Condition 2 in Theorem \ref{thm4} holds. 
Under Foliation C, for $ \gamma $	small enough, we set

$$ C'_\gamma=\min_{\partial\ov{\M}_\gamma} \left(4-(\ov{H}_{\partial\ov{\M}_\gamma}-q(\ov{e}_3))\sqrt{\dfrac{|\partial\ov{\M}_\gamma|_{\ov{g}}}{\pi}}\right)  $$  
where $ \Sigma_\gamma=\Psi_C(\Sigma,\gamma) $ and $ \ov{\Sigma}_\gamma $ is the  lift of $ \Sigma_\gamma$ onto $ \overline{M} $. $ \ov{\Sigma}_\gamma $ can also be treated as the horizontal cuts of the cylindrical end.\par

Recall that Foliation C is defined as :
\begin{align}
&\Psi_C : \Sigma\times [0,\bar \gamma] \to \M \mbox{ }\mbox{such that}\\ \nonumber
(1) & \Psi_C(p,0)=p \mbox{ }\mbox{for}\mbox{ }p\in\Sigma\\ \nonumber
(2) &  \Psi_C(p,\gamma)=u_0(p,-\frac{1}{\sqrt{\lambda}}\log\gamma)
\end{align}
where $ u_0 $ is the graph of $ \ov{\M} $ on the cylinder $ \Sigma\times\R $.\par

Then, by Theorem \ref{thm3} we have $ u_0 $ and $ ^{C_{\bar{ z}}}\nabla u_0 $ both uniformly converging to zero as $ z\to\infty $. Thus together with the fact that $ \ov{g }$ is the same as $ g $ on the level set of $ f $, and $ \lambda |\Sigma|_g\leq 4\pi$ (c.f. \cite[Lemma 9.7]{stable}), 
we have  $$ |\partial\ov{\M}_\gamma|_{\ov{g}}=|\partial {\M}_\gamma|_{ {g}} \to |\Sigma|_g\leq \frac{4\pi}{\lambda}\mbox{ }\mbox{when}\mbox{ }\gamma\to 0.$$ 
Moreover, it is proved in \cite{bk}  that  $ \overline{H}_{\partial\ov{\M}_\gamma} $ goes to zero. More specifically, $  \overline  {H}_{\partial\ov{\M}_\gamma} =<e_3, \overline{e}_3>H_{\partial {\M}_\gamma}  $. This is because
$ \overline  {g}=df_0^2+g $  does not change the metric $ g $ on the level sets of $ f_0 $ . It only  stretches lengths perpendicular to the level sets of $ f_0 $ by a factor of $ <e_3,\overline  {e}_3>=\frac{1}{\sqrt{1+|\nabla f_0|^2}} $. The formula $  \overline  {H}_{\partial\ov{\M}_\gamma} =<e_3, \overline{e}_3>H_{\partial {\M}_\gamma}  $ is then implied by the first variation formula for area. As we have proved that $ |\nabla f_0| $ blows up,  $ \overline{H}_{\partial\ov{\M}_\gamma} $ goes to zero as $\gamma\to 0 $.\par	

From \cite[Proposition 2]{yau} we know that $ |q|_{\ov{g}} $ is bounded near $ \Sigma $. If furthermore we suppose that the Condition 2 of Theorem \ref{thm4} holds, i.e. 
$|q|_{\ov{g}}<2\sqrt{\lambda} $ near $ \Sigma $, then it is straightforward to see that under this assumption,  there exists constant $ \gamma_0 $ s.t. $ C'_\gamma>0 $ for $ \forall\gamma\in(0,\gamma_0]  $. In all the context below,   $ \gamma_0 $ always refers to this definition.  \par
Denote $ P'^+_\gamma $ to be the $ L^2  $ orthogonal projection on the space of eigenvectors of positive eigenvalues of the corresponding $ \slashed{D} $ on $ \partial\ov{\M}_\gamma $, and denote $\HH'_\gamma=\{ \psi\in W_{-1}^{1,2}(\ov{\M}_\gamma),\linebreak P'^+_\gamma\psi=0  \}.$ Then similar arguments as in the previous section imply  the following inequality for spinors in $ \HH'_\gamma $:
\begin{equation}\label{bochnergamma}
\int_{\ov{\M}_{\gamma}}|D\psi|^2\geq	\frac{1}{2}\int_{\ov{\M}_{\gamma}}|\ov{\nabla}\psi|^2+\frac{1}{2}\sqrt{\dfrac{\pi}{|\partial\ov{\M}_\gamma|_{\ov{g}}    }} C'_\gamma \int_{\partial\ov{M_\gamma}}|\psi|^2\geq\frac{1}{2}\int_{\ov{\M}_{\gamma}}|\ov{\nabla}\psi|^2
\end{equation}

This implies a coercive estimate for spinors in $ \HH'_\gamma $. We then reach the same conclusion as in Proposition \ref{ppmt}:

\begin{proposition}\label{ppmtgamma}
	We suppose 
	$|q|_{\ov{g}}<2\sqrt{\lambda}$ near $ \Sigma $. Also we suppose $ \psi_0 $ is  a smooth spinor field  which is constant in some chart around infinity and $ P'^+_\gamma\psi_0=0 $ under Foliation C.  Then,  for $ \gamma<\gamma_0 $, there exists a unique $ \psi\in \HH'_\gamma $ s.t. 
	\begin{equation} 
	D(\psi_0+\psi)=0, P'^+_\gamma(\psi_0+\psi)=0
	\end{equation}
	Furthermore, we have the following lower bound for the  ADM mass of the initial data set:
	\begin{equation}\label{ppmtgamma eq}
	8\pi| \psi_0|^2 m\geq  \int_{\ov{\M}_{\gamma}}|\ov{\nabla}(\psi+\psi_0)|^2+ \sqrt{\dfrac{\pi}{|\partial\ov{\M}_\gamma|_{\ov{g}}    }} C'_\gamma \int_{\partial\ov{M_\gamma}}|\psi+\psi_0|^2
	\end{equation}
\end{proposition}

\subsection{A Penrose-like Inequality}

We shall now perform the last step in proving a Penrose-like inequality. We need following lemmas:
\begin{lemma}
	Assume $  X $ and $ Y $ are two Banach space such that $ X $ is continuously embedded into $ Y $. If $ X $ is  reflexive and $ \{x_{n}\} $ is a bounded sequence in $ X $, then there is a subsequence $ \{x_{n_k}\} $  weakly convergent to $ x\in X $ both in $ X $ and $ Y $.
\end{lemma}

\begin{proof}
	$ X $ is a reflexive Banach space, therefore bounded subsets are weakly precompact by Banach-Alaoglu theorem. The Eberlein-Smulian theorem implies that bounded subsets in X are weakly sequentially precompact, and therefore from $ \{x_n\} $ we can extract a subsequence $ \{x_{n_k}\} $ which is weakly converging to $ x\in X $. Let $ J $ to be the embedding operator from $ X $ to $ Y $. Then for $ f\in Y^* $, $ f\circ J\in X^* $ because $ J $ is continuous. Hence $ f(x_n)=(f\circ J)(x_n) $ converges to $ (f\circ J)(x)=f(x) $. Therefore, $  x_n $ converges weakly to x in Y.
\end{proof}  

\begin{lemma}
	Assume $ X $ and $ Z $ are two Banach space such that $ X $ is compactly embedded into $ Y $. If $ X $ is reflexive and $ \{x_{n}\} $ is bounded in $ X $, then there is a subsequence $ \{x_{n_k}\} $  convergent to $ x\in X $, weakly in $ X $ and strongly in $ Z $.
\end{lemma}

\begin{proof}
	For the same reason as above lemma, we can extract a subsequence $ \{x_{n_k}\} $ which is weakly converging to $ x\in X $. Because $ X $ is compactly embedded into $ Z $, thus there is an element  $ z\in Z $ s.t. there is a subsequence of $ \{x_n\} $, still denote as $ x_{n_k} $ , converges strongly to $ z $ in $ Z $. Thus $ \{x_{n_k}\} $ also converges weakly to $ z $ in $ Z $. For the same reason as above lemma, $ x_{n_k} $  converges weakly to $ x $ in $ Z $, thus $ x $ and $ z $ are equal.
\end{proof} 
\begin{lemma}
	Assume $\displaystyle \sigma_s=\sqrt{\dfrac{|\ov{\Sigma}_s|_{\ov{g}}}{\pi}}\inf_{f\in C_c^\infty(\ov{\Sigma}_{s})} \dfrac{|df|_{L^2(\ov{M}_{s})}^2}{|f|_{L^2(\ov{\Sigma}_{s})}^2}
	=\sqrt{\dfrac{|\ov{\Sigma}_s|_{\ov{g}}}{\pi}}\inf_{f\in W_{-1}^{1,2}(\ov{\Sigma}_{s})} \dfrac{|df|_{L^2(\ov{M}_{s})}^2}{|f|_{L^2(\ov{\Sigma}_{s})}^2}$, then $ \sigma_s>0$ for any $ s>0 $. Here f denotes function which is not identical zero.
\end{lemma}
\begin{proof}
	Assume $ \{f_i\} $ is a sequence of  functions minimizing $ \sigma_s $, and satisfies $ |f_i|_{L^2(\ov{\Sigma}_{s})}=1 $. Because $ W_{-1}^{1,2}(\ov{\M}_{s}) $ is compactly embedded into $ L^2(\ov{\Sigma}_{s} )$, by above lemma, we can find a subsequence $ \{f_{i_k}\} $ and $ f\in W_{-1}^{1,2} $, such that $ \{f_{i_k}\} $ converge to $ f $ weakly in $ W_{-1}^{1,2} $, strongly in $ L^2(\ov{\Sigma}_{s}) $. Thus $ |f|_{L^2(\ov{\Sigma}_{s})}=1 $, and  $ |df|_{L^2(\ov{\M}_{s})}\le \underline{lim}_{k\to\infty} |df_{i_k}|_{L^2(\ov{\M}_{s})}=\sigma_s$. Suppose that $ \sigma_s=0 $, then $ |f|_{L^2(\ov{\Sigma}_{s})}=1 $, while  $ |df|_{L^2(\ov{\M}_{s})}=0$. So $ f $ must be a constant on $ \ov{\M}_{s} $ while at the same time decay at infinity, which indicates that it can only be zero. This contradicts the fact that  $ |f|_{L^2(\ov{\Sigma}_{s})}=1 $.
\end{proof}

\begin{theorem}
	Under Foliation B, suppose $$\lambda \frac{C_1^2}{C_3^2}(1+\frac{2 C_1^4}{C_2^2C_3^2})<4, $$ where $ C_1, C_2, C_3 $ are constants appearing in the gradient estimates in Theorem \ref{thm3}. Then for  $ s\in (0,s_0] $, we have the following lower bound for the ADM mass of the initial data set $ (\M,g,K) $:
	\begin{equation}
	m \geq \frac{\sigma_s C_s}{ 2(C_s+\sigma_s)}\sqrt{\frac{|\ov{\Sigma}_{s}|_{\ov{g}}}{16\pi}}
	\end{equation}
\end{theorem}
\begin{proof}
	From Ineq.\eqref{ppmt eq} in Proposition \ref{ppmt}, and if for simplicity we assume $ \lim_{r\to\infty}|\psi_0|=1 $,
	we will  have:
	\begin{equation}\label{pen1}
	8\pi m\geq
	\int_{\ov{\M}_{s}}|\ov{\nabla}\psi_1|^2+\sqrt{\dfrac{\pi}{|\partial\ov{\M}_s|_{\ov{g}}    }} C_s \int_{\partial\ov{M_s}}|\psi_1|^2
	\end{equation}
	where $ \psi_1=\psi+\psi_0 $ is a spinor  such that  $ \lim_{r\to\infty}|\psi_1|=1 $.\par
	
	We now denote functions $ h=|\psi_1| $ and $ v=h-1 $. Then $ v\in W_{-1}^{1,2} $ as a function because of triangular formula. Because $ \ov{\nabla} $ is compatible with $ < , > $,  we have: 
	\begin{equation*}
	|\ov{\nabla}|\psi||^2=\frac{1}{4}|\psi|^{-2}|\ov{\nabla}(|\psi|^2)|^2
	=|\psi|^{-2}|<\ov{\nabla}\psi,\psi>|^2
	\leq|\ov{\nabla}\psi|^2
	\end{equation*}
	The above inequality  then implies:
	\begin{align}
	8\pi m&\geq
	\int_{\ov{\M}_{s}}|\ov{\nabla}\psi_1|^2+\sqrt{\dfrac{\pi}{|\partial\ov{\M}_s|_{\ov{g}}    }} C_s \int_{\partial\ov{M_s}}|\psi_1|^2\\
	\nonumber
	&\geq
	\int_{\ov{\M}_{s}}|\ov{\nabla}|\psi_1||^2+\sqrt{\dfrac{\pi}{|\partial\ov{\M}_s|_{\ov{g}}    }} C_s \int_{\partial\ov{M_s}}|\psi_1|^2\\
	\nonumber	
	&=
	\int_{\ov{\M}_{s}}|dv|^2+\sqrt{\dfrac{\pi}{|\partial\ov{\M}_s|_{\ov{g}}    }} C_s \int_{\partial\ov{M_s}}(1+v)^2
	\end{align}

	By Young's Inequality,
	\begin{equation}
	(1+v)^2\geq 1-\epsilon^{-1} +(1-\epsilon)v^2, \forall\epsilon>0
	\end{equation}
	From definition of $ \sigma_s $ we know that 
	
	\begin{equation}
	|dv|_{L^2(\ov{M}_{s})}^2 \geq \sigma_s \sqrt{\dfrac{\pi}{|\ov{\Sigma}_s|_{\ov{g}_s}}}  |v|_{L^2(\ov{\Sigma}_{s})}^2
	\end{equation} 
	
	Thus put together all the above inequalities and set $ \epsilon=1+ C^{-1}_s\sigma_s $, we get the following Penrose-like inequality for $ m $:
	\begin{align*}
	8\pi m
	&\geq\int_{\ov{\M}_{s}}|dv|^2+\sqrt{\dfrac{\pi}{|\partial\ov{\M}_s|_{\ov{g}}    }} C_s \int_{\partial\ov{M_s}}(1+v)^2\\
	&\geq|dv|_{L^2(\ov{M}_{s})}^2+
	\sqrt{\dfrac{\pi}{|\partial\ov{\M}_s|_{\ov{g}}    }}C_s\left( (1-\epsilon^{-1})|\partial\ov{\M}_s|_{\ov{g}}     +(1-\epsilon)|v|_{L^2(\partial\ov{\M}_s)}^2\right)\\
	&=\sqrt{\dfrac{\pi}{|\partial\ov{\M}_s|_{\ov{g}}    }}C_s (1-\epsilon^{-1})|\partial\ov{\M}_s|_{\ov{g}}    
	+|dv|_{L^2(\ov{M}_{s})}^2
	+\sqrt{\dfrac{\pi}{|\partial\ov{\M}_s|_{\ov{g}}    }}C_s(1-\epsilon)|v|_{L^2(\partial\ov{\M}_s)}^2\\
	&\geq\sqrt{\dfrac{\pi}{|\partial\ov{\M}_s|_{\ov{g}}    }}\left(C_s (1-\epsilon^{-1})|\partial\ov{\M}_s|_{\ov{g}}    
	+(\sigma_s+C_s(1-\epsilon))|v|_{L^2(\partial\ov{\M}_s)}^2\right)\\
	&= \frac{\sigma_s C_s}{ C_s+\sigma_s}\sqrt{\pi|\partial\ov{\M}_s|_{\ov{g}}    }	
	\end{align*}
\end{proof}

The similar reasoning holds for Foliation C if Condition 2 of Theorem \ref{thm4} holds, and thus we have a Penrose-like inequality following Proposition \ref{ppmtgamma}:

\begin{theorem}
	Under Foliation C, denote $\displaystyle \sigma'_\gamma=\sqrt{\dfrac{|\ov{\Sigma}_\gamma|_{\ov{g}}}{\pi}}\inf_{f\in C_c^\infty(\ov{\Sigma}_{\gamma})} \dfrac{|df|_{L^2(\ov{M}_{\gamma})}^2}{|f|_{L^2(\ov{\Sigma}_{\gamma})}^2}
	\linebreak=\sqrt{\dfrac{|\ov{\Sigma}_\gamma|_{\ov{g}}}{\pi}}\inf_{f\in W_{-1}^{1,2}(\ov{\Sigma}_{\gamma})} \dfrac{|df|_{L^2(\ov{M}_{\gamma})}^2}{|f|_{L^2(\ov{\Sigma}_{\gamma})}^2}$, then $ \sigma'_\gamma>0$ for any $ \gamma>0 $.
	If we suppose $|q|_{\ov{g}}<2\sqrt{\lambda}$ near $ \Sigma $, then for $ \gamma\in(0,\gamma_0] $, we have the following lower bound for the ADM mass of the initial data set $ (\M,g,K) $:
	\begin{equation}
	m \geq \frac{\sigma'_\gamma C'_\gamma}{ 2(C'_\gamma+\sigma'_\gamma)}\sqrt{\frac{|\ov{\Sigma}_{\gamma}|_{\ov{g}}}{16\pi}}
	\end{equation}
\end{theorem}

We now denote $ \theta_s= \dfrac{\sigma_s C_s}{2(C_s+\sigma_s)}\sqrt{\dfrac{|\partial\ov{\M}_s|_{\ov{g}}}{ |\partial\M|_g}}$, $ \theta'_\gamma= \dfrac{\sigma'_\gamma C'_\gamma}{2(C'_\gamma+\sigma'_\gamma)}\sqrt{\dfrac{|\partial\ov{\M}_\gamma|_{\ov{g}}}{ |\partial\M|_g}}$.

Denote 
\begin{equation}
\label{theta}
\theta=\left\{
\begin{aligned}
\sup_{s\in(0,s_0)}\theta_s &   &  \mbox{if}\mbox{ }\mbox{Condition}\mbox{ }\mbox{1}\mbox{ }\mbox{holds,} \\
\sup_{\gamma\in(0,\gamma_0)}\theta'_\gamma &   & \mbox{if}\mbox{ }\mbox{Condition}\mbox{ }\mbox{2}\mbox{ }\mbox{holds.}
\end{aligned}
\right.
\end{equation}

Then either Condition 1 or 2 of Theorem \ref{thm4} implies that there is a constant $ \theta>0 $ such that

\begin{equation}
m\ge \theta\sqrt{\dfrac{|\Sigma|_g}{16\pi}}
\end{equation}

\subsection{Schwarzschild Time Symmetric Case}

Now we calculate the value of $ \theta $ in Schwarzschild  time symmetric case. Suppose now our initial data set is the Schwarzschild  time symmetric slice, and for simplicity we assume $ m=1 $. Following J. Metzger's construction procedure in \cite[Theorem 3.1]{metzger}, each solution $ f_t $ of   $ \J[f_t]=tf_t $ with Dirichlet boundary condition under spherically symmetric setting will still be spherically symmetric. Thus the Jang's slice $ \ov{\M}=graph f_0 $, which is the limiting manifold of   $ graph f_t $ when $ t\to 0 $, is also spherically symmetric. It is then straightforward to find the only spherically symmetric solution $ f_0 $ of $ \J[f_0]=0 $ on $ \M $ such that it blows up at $ r=2 $, and decays to zero at infinity. By solving the ODE, we know $ f_0  $ satisfies: 
$$ f_0'(r)=-\sqrt{\dfrac{16r}{(r-2)(r^4-16)}} $$

We  use the coordinate system $ \{ r,\theta,\phi \} $ for both $ \M $ and $ \ov{\M} $. First of all, $ \dfrac{\partial}{\partial r}\theta^+(S_r)|_{r=2}>0 $, thus the horizon $ \{r=2\} $ is outermost strictly stable MOTS.  Also under this coordinate system, the distance function $ \tau $ to the horizon $ \{ r=2 \} $ on $ \M $  is of order $ \sqrt{r-2} $, thus $$ \log \tau=\dfrac{1}{2}\log (r-2) +O(1). $$ Together with the fact that  $$ f_0(r)=-\log(r-2)+O(1)=-2\log\tau+O(1) ,$$ by Theorem 1 the principal eigenvalue of the horizon $ \{r=2\} $ in $ \M $ must be $ \lambda =\dfrac{1}{4}$.\par

Now the induced metric $ \ov{g} $ on $ \ov{\M} $ is
$$\ov{g}=g+df_0^2=\frac{r^5}{(r-2)(r^4-16)}dr^2+r^2d\theta^2+r^2\sin^2\theta d\phi^2 ,$$ and it can be calculated that:
\begin{align*}
|\ov{\Sigma}_r|_{\ov{g}} =&4\pi r^2\\
\ov{H}_{\ov{\Sigma}_r}=&2\sqrt{\dfrac{(r-2)(r^4-16)}{r^7}}  \\ q(\ov{e}_3)=&-\dfrac{32}{\sqrt{r^7(r^3+2r^2+4r+8)}} 	
\end{align*}

Thus 
\begin{align*}
C_r=& 4-(\ov{H}_{\partial\ov{\M}_r}-q(\ov{e}_3))\sqrt{\dfrac{|\partial\ov{\M}_r|_{\ov{g}}}{\pi}}\\
=&4-4 \sqrt{\dfrac{(r-2)(r^4-16)}{r^5}} -\dfrac{64 }{\sqrt{r^5(r^3+2r^2+4r+8)}}  
\end{align*}
$ C_r $ is positive for $ \forall r\geq 2 $.\par

The following lemma gives the value of $ \sigma_r $:

\begin{lemma}\label{sigma}
	Suppose $ (M,g) $ is a 3-dim spherically symmetric Riemannian manifold equipped with $ g=F^2(r)(r)dr^2+r^2d\theta^2+r^2\sin^2\theta d\phi^2$ and boundary $ \partial M =\{r=r_0\}$. Then for any  $ f\in C_c^\infty(M) $  such that $ |f|_{L^2(\partial M)}=1 $, we have 
	\begin{equation}
	|df|^2_{L^2(M)}\geq\frac{1}{r_0^2 \int_{r_0}^{\infty}\frac{F(r)}{r^2}}
	\end{equation} 
	and the equality holds if and only if $ f $ is the spherically symmetric harmonic solution on $ M $.
	
\end{lemma}

\begin{proof}
	The inequality can be proved by Fubini's theorem and Cauchy-Schwarz Inequality:	
	{\allowdisplaybreaks	\begin{align*}
		\int_M |df|^2=& \int d\theta\int d\phi \int_{r_0}^{\infty}
		|df|^2	\sqrt{det g}dr\\
		\geq& \int d\theta\int d\phi \int_{r_0}^{\infty}
		\frac{1}{F^2(r)} (\partial_r f(r,\theta,\phi))^2	\sqrt{det g}dr\\
		=&  \int\sin\theta d\theta d\phi  \int_{r_0}^{\infty} (\partial_r f(r,\theta,\phi))^2 \frac{r^2}{F(r)}   \\
		\geq&    \iint \sin\theta d\theta d\phi \frac{(\int_{r_0}^{\infty}\partial_r f(r,\theta,\phi))^2}{\int_{r_0}^{\infty}\frac{F(r)}{r^2}}    \\	
		=&   \frac{1}{\int_{r_0}^{\infty}\frac{F(r)}{r^2}}  \iint   f(r_0,\theta,\phi)^2\sin\theta d\theta d\phi  \\	
		=& \frac{1}{\int_{r_0}^{\infty}\frac{F(r)}{r^2}} \frac{|f|_{L^2(\partial M)}^2  }{r_0^2}   	
		=\frac{1}{r_0^2\int_{r_0}^{\infty}\frac{F(r)}{r^2}} 	
		\end{align*}}
	
	The equality holds if and only if $ f=f(r) $ and $ f'(r)=C \dfrac{F(r)}{r^2} $	. It can be checked  that if this happens then $ \Delta_M f=0 $.	
\end{proof}

It is then straightforward to calculate the value of $ \sigma_r $ by Lemma \ref{sigma}:
\begin{align*}
\sigma_r=&\sqrt{\dfrac{|\ov{\Sigma}_r|_{\ov{g} }}{\pi}}\inf_{f\in C_c^\infty(\ov{\Sigma}_{r})} \dfrac{|df|_{L^2(\ov{M}_{r})}^2}{|f|_{L^2(\ov{\Sigma}_{r})}^2}\\
=&	\frac{2 }{r\int_{r }^{\infty}\sqrt{\frac{x}{(x-2)(x^4-16)}}dx}
\end{align*}

Then

\begin{align*}
\theta_r=& \dfrac{\sigma_r C_r}{2(C_r+\sigma_r)}\sqrt{\dfrac{|\partial\ov{\M}_r|_{\ov{g}}}{|\partial\M|_g}}\\
=&\frac{r }{4(\sigma_r^{-1}+C_r^{-1})}\\
=& \frac{r}{2r\int_{r }^{\infty}\sqrt{\frac{x}{(x-2)(x^4-16)}}dx+\frac{1}{1-  \sqrt{\frac{(r-2)(r^4-16)}{r^5}} -\frac{16 }{\sqrt{r^5(r^3+2r^2+4r+8)}} }}
\end{align*}

$ \theta_r $ is monotonically increasing in $ r $, $ \lim_{r\to2}\theta_r=0 $, $ \lim_{r\to\infty}\theta_r=1 $. Although $ \lim_{r\to2}\theta_r=0 $ is not a good property, $ \theta_r $  quickly becomes significantly non-zero when it leaves away from the horizon $\{r=2\}  $. We list a few numerical results here: 
\begin{table}[h!]
	\begin{center}
		\caption{Numerical Values of $ \theta_r $}
		\label{tab:table1}
		\begin{tabular}{l| r} 
			
			$r$ & $\theta_r$ \\
			\hline
			2.001 & 0.3198 \\
			2.01 & 0.3922 \\
			2.1 & 0.5084 \\
			2.5 & 0.6466 \\
			3.0 & 0.7292  \\
		\end{tabular}
	\end{center}
\end{table}

Thus for example if we look at the domain $ U=\{r\in(2,2.1)\} $, then $ \theta=\sup_U \theta_r\approx 0.5 $.\par

However, the property that $ \lim_{r\to2}\theta_r=0 $ is still not a good one. This property means that when we approach the cylindrical end we will lose more and more information, which contradicts our intuition. The following calculation shows that most of the information loss happens at the last step: when we apply Young's Inequality together with capacity.\par  

We first calculate the Green's function of the Dirac operator on $ \M $.
\begin{lemma}\label{spin}
	Suppose $ (M,g) $ is a 3-dim spherically symmetric Riemannian manifold equipped with $ g=F^2(r)dr^2+r^2d\theta^2+r^2\sin^2\theta d\phi^2$ and boundary $ \partial M =\{r=r_0\}$. Denote $ D $ to be the Dirac operator on $ M $. Then the following:
	
	\begin{align*}
	\psi=c_1 e^{\frac{i\phi}{2}}\begin{bmatrix}
	\cos\frac{\theta}{2}  \\ -\sin\frac{\theta}{2} 
	\end{bmatrix}h(r)
	+c_2 e^{-\frac{i\phi}{2}}\begin{bmatrix}
	\sin\frac{\theta}{2}  \\ \cos\frac{\theta}{2} 
	\end{bmatrix}h(r)
	\end{align*}
	is a solution of $ D\psi=0 $.
	Here $ c_1,c_2 $ are complex numbers, and 
	$$ h(r)=C e^{-\int_{r}^{\infty}\frac{F(s)-1}{s}ds}$$
	where $ C $ is a real constant. Furthermore, we have 
	$ |\psi|=(|c_1|^2+|c_2|^2)h(r)$, and  
	$ \slashed{D}\psi=-\frac{1}{r_0 }\psi $ on the boundary $ \partial M =\{r=r_0\}$
	
\end{lemma}

\begin{proof}
	Denote the associated orthonormal frame $ \{e_i\} $ and coframe $ \{\omega_i\} $ on $ M $ by:
	\begin{align*}
	e_1=\frac{1}{r}\frac{\partial }{\partial\theta},\mbox{   }\mbox{   }\mbox{   }\mbox{   }\mbox{   } e_2=\frac{1}{r\sin\theta}\frac{\partial }{\partial \phi},\mbox{   }\mbox{   }\mbox{   }\mbox{   }\mbox{   } 
	e_3=\frac{1}{F(r)}\frac{\partial }{\partial r}.
	\end{align*}
	and 
	\begin{align*}
	\omega_1=rd\theta,  \mbox{   }\mbox{   }\mbox{   }\mbox{   }\mbox{   }  \omega_2=r\sin\theta d\phi, \mbox{   }\mbox{   }\mbox{   }\mbox{   }\mbox{   }
	\omega_3=F(r)dr.
	\end{align*}
	
	The connection 1-form $ \{\omega_{ij}\} $ is given by $ d\omega_i=-\omega_{ij}\wedge\omega_j $. Calculation shows:
	\begin{align*}
	\omega_{31}=-\frac{1}{rF(r)}\omega_1, \mbox{   }\mbox{   }\mbox{   }\mbox{   }\mbox{   }  \omega_{32}=-\frac{1}{rF(r)}\omega_2,  \mbox{   }\mbox{   }\mbox{   }\mbox{   }\mbox{   }  \omega_{12}=-\frac{\cot\theta}{r}\omega_2.
	\end{align*}
	
	The spin connection is given by:
	\begin{equation*}
	{\nabla}=d-\frac{1}{4}\omega_{ij}\otimes e_i \cdot e_j  =d-\frac{1}{2}\omega_{12}\otimes e_1 \cdot e_2  -\frac{1}{2}\omega_{13}\otimes e_1 \cdot e_3 \cdot-\frac{1}{2}\omega_{23}\otimes e_2 \cdot e_3  
	\end{equation*}
	
	Fix the following Pauli matrix throughout this section
	\begin{align*}
	e_1\rightarrow\begin{bmatrix}
	&i    \\
	i\mbox{ }&   
	\end{bmatrix}, 
	\mbox{   }\mbox{   }\mbox{   }\mbox{   }\mbox{   }
	e_2\rightarrow\begin{bmatrix}
	&1    \\
	-1&  
	\end{bmatrix},
	\mbox{   }\mbox{   }\mbox{   }\mbox{   }\mbox{   }
	e_3\rightarrow\begin{bmatrix}
	i       &    \\
	& -i 
	\end{bmatrix}.
	\end{align*}

	Then the Dirac operator is: 
	\begin{align*}
	D =&e_1\cdot{\nabla}_{e_1}+e_2\cdot{\nabla}_{e_2}+e_3\cdot{\nabla}_{e_3}\\
	=&e_3\cdot\frac{1}{F(r)}\frac{\partial }{\partial r}
	+e_1\cdot\frac{1}{r}\frac{\partial }{\partial\theta}
	+ e_2\cdot\frac{1}{r\sin\theta}\frac{\partial }{\partial \phi} \\
	&+\frac{1}{2r F(r)}(e_1\cdot e_3\cdot e_1+e_2\cdot e_3\cdot e_2)
	+\frac{\cot\theta}{2r}e_2\cdot e_1\cdot e_2	
	\end{align*}
	
	Then by separation of variable, we can find a solution of $ D\psi=0 $:
	\begin{align*}
	\psi=c_1 e^{\frac{i\phi}{2}}\begin{bmatrix}
	\cos\frac{\theta}{2}  \\ -\sin\frac{\theta}{2} 
	\end{bmatrix}h(r)
	+c_2 e^{-\frac{i\phi}{2}}\begin{bmatrix}
	\sin\frac{\theta}{2}  \\ \cos\frac{\theta}{2} 
	\end{bmatrix}h(r)
	\end{align*}
	
	where $ c_1,c_2 $ are complex numbers, and $ h(r) $ is a real-value function satisfies:
	\begin{equation*}
	rh'(r)=(F(r)-1)h(r)
	\end{equation*}
	then 
	$$ h(r)=C e^{-\int_{r}^{\infty}\frac{F(s)-1}{s}ds}$$
	It can be checked that $ |\psi|=(|c_1|^2+|c_2|^2)h(r)$.  Also
	on the boundary $ \partial M=S_{r_0} $, the orthonormal frame will be $ \{e_1,e_2\} $, and $ \nu=e_3 $. It is then straightforward to get the following result on $ S_{r_0} $:
	\begin{equation*}\psi=e_3\cdot e_1\cdot {\nabla}^{S_{r_0} }_{e_1}\psi+e_3\cdot e_2\cdot {\nabla}^{S_{r_0} }_{e_2}\psi=\frac{1}{r_0 }\psi
	\end{equation*}
	where $ {\nabla}^{S_{r_0}}  $  is the spin connection on the boundary.

	Thus $ \slashed{D}\psi=-\frac{1}{r_0 }\psi $.
\end{proof}

We apply the above lemma to the Jang's slice $ \ov{\M} $, then
$ F(r)=\sqrt{\dfrac{r^5}{(r-2)(r^4-16)}} $.
If we require that $ h(r)\to 1 $ when $ r\to\infty $,
then 
$$ h(r)=e^{-\int_{r}^{\infty}\frac{F(s)-1}{s}ds}$$
Under this setting $ \psi $ is constant at infinity. It can also be checked that $ h(r)\to 0 $ when $ r\to 2 $, and $ h(r)\to 1 $ when $ r\to\infty $. Thus if we set $ |c_1|^2+|c_2|^2=1 $, then $ |\psi|\to 1 $ when $ r\to \infty $.\par

Then 
\begin{align*}
&\lim_{r\to2}(\int_{\ov{\M}_{r}}|dh|^2+\sqrt{\dfrac{\pi}{|\partial\ov{\M}_r|_{\ov{g}}    }} C_r \int_{\partial\ov{M_r}}h^2)\\
=&\lim_{r\to2}\int_{\ov{\M}_{r}}|dh|^2\\
=&\iint d\theta d\phi \int_{2}^{\infty} \frac{1}{F(r)} (h'(r))^2\sqrt{det g}dr\\
=&4\pi\int_{2}^{\infty} \frac{r^2}{\sqrt{F(r)}}( h'(r))^2dr
\end{align*}
By numerical calculation $ \int_{2}^{\infty} \frac{r^2}{\sqrt{F(r)}}( h'(r))^2dr\approx 0.6795 $. Although it is not sharp (sharp value should be 2), it is not zero, compared to the fact that $ \lim_{r\to2}\theta_r=0 $.\par 

There are some other information loss. One of them happens when we are applying the Ineq.\eqref{isometric}:
\begin{equation}
|\ov{\nabla}\psi|^2+|q|^2|\psi|^2+q(|\psi|^2)\ge 0
\end{equation}

Another happens when we are using the inequality $|\ov{\nabla}|\psi|| \leq|\ov{\nabla}\psi| $:

One can  calculate that: 
\begin{equation*}
\lim_{r\to2}(\int_{\ov{\M}_{r}}|\ov{\nabla}\psi|^2+\sqrt{\dfrac{\pi}{|\partial\ov{\M}_r|_{\ov{g}}    }} C_r \int_{\partial\ov{M_r}}|\psi|^2)\approx 4\pi* 1.0193
\end{equation*}
which is a little sharper than before.\par

\subsection{Conclusion}
We prove a Penrose-like inequality $ m\ge \theta\sqrt{\dfrac{|\Sigma|_g}{16\pi}} $ under two different conditions on a one-end asymptotically flat 3-dim initial data set $ (\M,g,k) $ with boundary $ \partial\M $, which is a connected   compact  strictly stable outermost MOTS. We are able to keep most of the information about the cylindrical end after solving Dirac equation Eq.\eqref{spin pde}. The most serious information loss happens at the last step, when we are using capacity to extract an area term. This is caused by the essential difference between  the Green's function  for    Laplacian and the Green's function  for   Dirac operator.

\providecommand{\href}[2]{#2}

\address{
	Columbia University,\\
	 New York, USA\\
	\email{wy2227@columbia.edu}
}

\end{document}